\documentclass{IEEEtran}

\usepackage{amsmath}
\usepackage{mathtools}
\usepackage{graphicx}
\usepackage{enumerate}
\usepackage{url}
\usepackage{subfigure}
\usepackage{algorithm}
\usepackage{algpseudocode}
\usepackage{stfloats}
\usepackage{lineno}
\usepackage{hyperref}
\usepackage{bm}
\usepackage{bbm}
\usepackage{amssymb}
\usepackage{enumerate}
\usepackage{xcolor}
\usepackage{float}
\usepackage{cite}
\usepackage{tabularx}
\usepackage{tabu}
\usepackage{multicol}
\usepackage{multirow}
\usepackage{colortbl,booktabs,threeparttable}
\usepackage{dcolumn}
\usepackage{flushend}
\usepackage{soul}
\usepackage{CJKutf8}
\usepackage{ragged2e}

\graphicspath{{Figures/}}

\renewcommand{\vec}[1]{\bm #1}
\newcommand{\rvec}[1]{\mathbf{#1}}
\newcommand{\mat}[1]{\bm #1}
\newcommand{\rmat}[1]{\mathbf{#1}}

\newcommand{\rvech}[1]{\hat{\mathbf{#1}}}
\renewcommand{\math}[1]{\hat{\bm #1}}

\newcommand{\vecul}[1]{\underline{\vec{#1}}}
\newcommand{\rvecul}[1]{\underline{\rvec{#1}}}
\newcommand{\matul}[1]{\underline{\mat{#1}}}
\newcommand{\matuul}[1]{\underline{\underline{\mat{#1}}}}

\renewcommand{\cal}[1]{\mathcal{#1}}
\newcommand{\bfit}[1]{\textbf{\textit{#1}}}
\newcommand{\tr}[1]{\operatorname{Tr}\left[{#1}\right]}
\newcommand{\Tr}{\operatorname{Tr}}
\renewcommand{\Re}{\operatorname{Re}}
\renewcommand{\Im}{\operatorname{Im}}

\newcommand{\R}{\mathbb{R}}
\newcommand{\C}{\mathbb{C}}
\newcommand{\E}{\mathbb{E}}

\renewcommand{\P}{\mathbb{P}}
\newcommand{\Ph}{\hat{\mathbb{P}}}

\newcommand{\w}{\omega}

\renewcommand{\th}{\text{th}}

\newcommand{\T}{\mathsf{T}}
\renewcommand{\H}{\mathsf{H}}
\renewcommand{\st}{\text{s.t.}}

\renewcommand{\ul}{\underline}

\newcommand{\defeq}{\coloneqq}
\newcommand{\stp}{\hfill $\square$}

\newcommand{\tabincell}[2]{\begin{tabular}{@{}#1@{}}#2\end{tabular}}

\definecolor{hl-bg-color}{RGB}{255,255,215}
\sethlcolor{hl-bg-color} 
\definecolor{new-magenta}{RGB}{255,0,255}
\soulregister\cite7
\soulregister\citep7
\soulregister\citet7
\soulregister\ref7
\soulregister\eqref7
\newcommand*{\HIGHLIGHT}{}
\ifdefined\HIGHLIGHT
\newcommand{\red}[1]{{\color{red}{#1}}}
\newcommand{\blue}[1]{{\color{blue}{#1}}}

\else

\newcommand{\red}[1]{{#1}}
\newcommand{\blue}[1]{{#1}}

\fi

\newtheorem{theorem}{{Theorem}}
\newtheorem{proposition}{{Proposition}}
\newtheorem{corollary}{{Corollary}}
\newtheorem{lemma}{{Lemma}}
\newtheorem{claim}{{Claim}}
\newtheorem{fact}{{Fact}}

\newtheorem{definition}{{Definition}}

\newtheorem{example}{{Example}}

\begin{document}
\newpage
\title{Distributionally Robust Receive Combining}

\author{Shixiong Wang,~
        Wei Dai,
        and Geoffrey Ye Li,~\IEEEmembership{Fellow,~IEEE}
\thanks{S. Wang, W. Dai, and G. Li are with the Department of Electrical and Electronic Engineering, Imperial College London, London SW7 2AZ, United Kingdom (E-mail: s.wang@u.nus.edu; wei.dai1@imperial.ac.uk; geoffrey.li@imperial.ac.uk).
}
\thanks{This work is supported by the UK Department for Science, Innovation
and Technology under the Future Open Networks Research Challenge project
TUDOR (Towards Ubiquitous 3D Open Resilient Network). 
}
}

\maketitle

\begin{abstract}
This article investigates signal estimation in wireless transmission (i.e., receive combining) from the perspective of statistical machine learning, where the transmit signals may be from an integrated sensing and communication system; that is, 1) signals may be not only discrete constellation points but also arbitrary complex values; 2) signals may be spatially correlated. Particular attention is paid to handling various uncertainties such as the uncertainty of the transmit signal covariance, the uncertainty of the channel matrix, the uncertainty of the channel noise covariance, the existence of channel impulse noises, the non-ideality of the power amplifiers, and the limited sample size of pilots. To proceed, a distributionally robust receive combining framework that is insensitive to the above uncertainties is proposed, which reveals that channel estimation is not a necessary operation. For optimal linear estimation, the proposed framework includes several existing combiners as special cases such as diagonal loading and eigenvalue thresholding. For optimal nonlinear estimation, estimators are limited in reproducing kernel Hilbert spaces and neural network function spaces, and corresponding uncertainty-aware solutions (e.g., kernelized diagonal loading) are derived. In addition, we prove that the ridge and kernel ridge regression methods in machine learning are distributionally robust against diagonal perturbation in feature covariance.
\end{abstract}


\section{Introduction}\label{sec:intro}

\IEEEPARstart{I}{n} wireless transmission, detection and estimation of transmitted signals is of high importance, and combining at array receivers serves as a key signal-processing technique to suppress interference and environmental noises. The earliest beamforming solutions rely on the use of phase shifters (e.g., phased arrays) to steer and shape wave lobes, while advanced combining methods allow the employment of digital signal processing units, which introduce additional structural freedom (e.g., fully digital, hybrid, nonlinear, wideband) in combiner design and significant performance improvement in signal recovery \cite{lo1991nonlinear,yang2015fifty,elbir2023twenty}.

In traditional communication systems, transmitted signals are discrete points from constellations. Therefore, signal recovery, commonly referred to as \textit{signal detection}, can be cast into a classification problem from the perspective of statistical machine learning, and the number of candidate classes is determined by the number of points in the employed constellation. Research in this stream includes, e.g., \cite{chen2008adaptive,chen2008symmetric,navia2010approximate,neinavaie2020lossless,liao2023deep,awan2023robust} 
as well as references therein, and the performance measure for signal detection is usually the misclassification rate (i.e., symbol error rate); representative algorithms encompass the maximum likelihood detector, the sphere decoding, etc. 
In another research stream, the signal recovery performance is evaluated using mean-squared errors (cf., signal-to-interference-plus-noise ratio), and the resultant signal recovery problem is commonly known as \textit{signal estimation}, which can be considered as a regression problem from the perspective of statistical machine learning. By comparing the estimated symbols with the constellation points afterward, the detection of discrete symbols can be realized. For this case, till now, typical combining solutions include zero-forcing receivers, Wiener receivers (i.e., linear minimum mean-squared error receivers), Capon receivers (i.e., minimum variance distortionless response receivers), and nonlinear receivers such as neural-network receivers \cite{ye2017power,he2020model,van2022transfer}. On the basis of these canonical approaches, variants such as robust beamformers working against the limited size of pilot samples and the uncertainty in steering vectors \cite{li2003robust,lorenz2005robust,zhang2015robust,li2017distributionally,huang2022robust,huang2023robust}
have also been intensively reported; among these robust solutions, the diagonal loading method \cite{cox1987robust}, \cite[Eq.~(11)]{lorenz2005robust} and the eigenvalue thresholding method \cite{harmanci2000relationships}, \cite[Eq.~(12)]{lorenz2005robust} are popular due to their excellent balance between practical performance and technical simplicity.

Different from traditional paradigms, in emerging communication systems, e.g., integrated sensing and communication (ISAC) systems, transmitted signals may be arbitrary complex values and spatially correlated \cite{liu2018toward,zhang2021overview,xiong2023fundamental}. As a result, mean-squared error is a preferred performance measure to investigate the \textit{receive combining and estimation} problem of wireless signals, which is, therefore, the focus of this article.

Although a large body of problems have been attacked in the area, the following signal-processing problems of combining and estimation 
in wireless transmission remain unsolved.
\begin{enumerate}
    \item What is the relation between the signal-model-based approaches (e.g., Wiener and Capon receivers) and the data-driven approaches (e.g., deep-learning receivers)? In other words, how can we build a mathematically unified modeling framework to interpret all the existing digital receive combiners?
    
    \item In addition to the limited pilot size and the uncertainty in steering vectors, there exist other uncertainties in the signal model: the uncertainty of the transmit signal covariance, the uncertainty of the communication channel matrix, the uncertainty of the channel noise covariance, the presence of channel impulse noises (i.e., outliers), and the non-ideality of the power amplifiers. Therefore, how can we handle all these types of uncertainties in a unified solution framework?
    
    \item Existing literature mainly studied the robustness theory of linear beamformers against limited pilot size and the uncertainty in steering vectors  \cite{li2003robust,lorenz2005robust,zhang2015robust,li2017distributionally,huang2022robust,huang2023robust}. However, how can we develop the theory of robust nonlinear combiners against all the aforementioned uncertainties?
\end{enumerate}

To this end, this article designs a unified modeling and solution framework for receive combining of wireless signals, in consideration of the scarcity of the pilot data and the different uncertainties in the signal model.

\subsection{Contributions}
The contributions of this article can be summarized from the aspects of machine learning theory and wireless transmission theory.

In terms of machine learning theory, we give a justification of the popular ridge regression and kernel ridge regression (i.e., quadratic loss function plus squared-$F$-norm regularization) from the perspective of distributional robustness against diagonal perturbation in feature covariance, which enriches the theory of trustworthy machine learning; see Theorems \ref{thm:ridge-regression} and \ref{thm:kernel-ridge}, as well as Corollaries \ref{cor:adv-learning-linear} and \ref{cor:adv-learning-linear-RKHS}. 

In terms of wireless transmission theory, the contributions are outlined below.
\begin{enumerate}
    \item We build a fundamentally theoretical framework for receive combining from the perspective of statistical machine learning. In addition to the linear estimation methods, nonlinear approaches (i.e., nonlinear combining) are also discussed in reproducing kernel Hilbert spaces and neural network function spaces. In particular, we reveal that channel estimation is not a necessary operation in receive combining. For details, see Subsection \ref{subsec:nonlin-estimation}.

    \item The presented framework is particularly developed from the perspective of distributional robustness which can therefore combat the limited size of pilot data and several types of uncertainties in the wireless signal model such as the uncertainty in the transmit power matrix, the uncertainty in the communication channel matrix, the existence of channel impulse noises (i.e., outliers), the uncertainty in the covariance matrix of channel noises, the non-ideality of the power amplifiers, etc. For details, see Subsection \ref{subsec:dist-uncertainty}, and the technical developments in Sections \ref{sec:lin-esti} and \ref{sec:nonlin-esti}.

    \item  Existing methods such as diagonal loading and eigenvalue thresholding are proven to be distributionally robust against the limited pilot size and all the aforementioned uncertainties in the wireless signal model. Extensions of diagonal loading and eigenvalue thresholding are proposed as well. Moreover, the kernelized diagonal loading and the kernelized eigenvalue thresholding methods are put forward for nonlinear estimation cases. For details, see Corollary \ref{cor:solutions-under-moments}, Examples \ref{exam:opt-estimator-RKHS} and \ref{exam:opt-estimator-RKHS-eigen-thres}, and Subsections \ref{subsec:BF-nontrivial-sets}.

    \item The distributionally robust receive combining and signal estimation problems across multiple frames, where channel conditions may change, are also investigated. For details, see Subsections \ref{subsec:BF-multi-frame} and \ref{subsubsec:multi-frame-RKHS}.
\end{enumerate}

\subsection{Notations}\label{subsec:notation}
The $N$-dimensional real (coordinate) space and complex (coordinate) space are denoted as $\R^N$ and $\C^N$, respectively. Lowercase symbols (e.g., $\vec x$) denote vectors (column by default) and uppercase ones (e.g., $\mat X$) denote matrices. We use the Roman font for random quantities (e.g., $\rvec x, \rmat X$) and the italic font for deterministic quantities (e.g., $\vec x, \mat X$). Let $\Re \vec X$ be the real part of a complex quantity $\mat X$ (a vector or matrix) and $\Im \mat X$ be the imaginary part of $\mat X$.
For a vector $\vec x \in \C^N$, let
\[
\vecul x \defeq
\left[
\begin{array}{cc}
    \Re \vec x \\
    \Im \vec x
\end{array}
\right] \in \R^{2N}
\]
be the real-space representation of $\vec x$; for a matrix $\vec H \in \C^{N \times M}$, let
\[
\matul H \defeq
\left[
\begin{array}{cc}
    \Re \mat H \\
    \Im \mat H
\end{array}
\right]
,~~~~~
\ul{\matul H} \defeq
\left[
\begin{array}{cc}
    \Re \mat H & -\Im \mat H\\
    \Im \mat H & \Re \mat H
\end{array}
\right]
\]
be the real-space representations of $\mat H$ where $\matul H \in \R^{2N \times M}$ and $\ul{\matul H} \in \R^{2N \times 2M}$. The running index set induced by an integer $N$ is defined as $[N] \defeq \{1,2,\ldots,N\}$. To concatenate matrices and vectors, MATLAB notations are used: i.e., $[\mat A,~\mat B]$ for row stacking and $[\mat A;~\mat B]$ for column stacking. We let $\mat \Gamma_M \defeq [\mat I_M,~ \mat J_M] \in \C^{M \times 2M}$ where $\mat I_M$ denotes the $M$-dimensional identity matrix, $\mat J_M \defeq j \cdot \mat I_M$, and $j$ denotes the imaginary unit. Let $\cal N (\vec \mu, \mat \Sigma)$ denote a real Gaussian distribution with mean $\vec \mu$ and covariance $\mat \Sigma$. We use $\cal{CN}(\vec s, \mat P, \mat C)$ to denote a complex Gaussian distribution with mean $\vec s$, covariance $\mat P$, and pseudo-covariance $\mat C$; if $\mat C$ is not specified, we imply $\mat C = \mat 0$. 

\section{Preliminaries}\label{sec:prelimilary}

We review two popular structured representation methods of nonlinear functions $\vec \phi: \R^{N} \to \R^{M}$. More details can be seen in Appendix \ref{append:func-representation}.

\subsection{Reproducing Kernel Hilbert Spaces}\label{subsec:pre-RKHS}
A reproducing kernel Hilbert space (RKHS) $\cal H$ induced by the kernel function $\ker: \R^N \times \R^N \to \R$ and a collection of points $\{\vec x_1, \vec x_2, \ldots, \vec x_L\} \subset \R^N$ is a set of functions from $\R^N$ to $\R$; $L$ may be infinite. Every function $\phi: \R^N \to \R$ in the functional space $\cal H$ can be represented by a linear combination \cite[p.~539;~Chap.~14]{murphy2012machine}
\begin{equation}\label{eq:non-lin-func-RKHS}
\phi(\vec x) = \sum^L_{i = 1} \w_i \cdot \ker(\vec x, \vec x_i),~\forall \vec x \in \R^N
\end{equation}
where $\{\w_i\}_{i \in [L]}$ are the combination weights; $\w_i \in \R$ for every $i \in [L]$. The matrix form of \eqref{eq:non-lin-func-RKHS} for $M$-multiple functions are
\begin{equation}\label{eq:non-lin-func-multi-RKHS}
\vec \phi(\vec x) \defeq 
\left[
\begin{array}{ccccccc}
    \phi_1(\vec x) \\
    \phi_2(\vec x) \\
    \vdots \\
    \phi_M(\vec x)
\end{array}
\right] = 
\mat W \cdot 
\vec \varphi(\vec x) \defeq
\left[
\begin{array}{c}
    \vec \w_1 \\
    \vec \w_2 \\
    \vdots \\
    \vec \w_M
\end{array}
\right] \cdot 
\vec \varphi(\vec x),
\end{equation}
where $\vec \w_1,\vec \w_2,\ldots,\vec \w_M \in \R^{L}$ are weight row-vectors for functions $\phi_1(\vec x), \phi_2(\vec x), \ldots, \phi_M(\vec x)$, respectively, and
\begin{equation}\label{eq:RKHS-defs}
\mat W \defeq
\left[
\begin{array}{c}
    \vec \w_1 \\
    \vec \w_2 \\
    \vdots \\
    \vec \w_M
\end{array}
\right] \in \R^{M \times L},~~~
\vec \varphi(\vec x) \defeq
\left[
\begin{array}{c}
    \ker(\vec x, \vec x_1) \\
    \ker(\vec x, \vec x_2) \\
    \vdots \\
    \ker(\vec x, \vec x_L)
\end{array}
\right].
\end{equation}
Since a kernel function is pre-designed (i.e., fixed) for an RKHS $\cal H$, \eqref{eq:non-lin-func-multi-RKHS} suggests a $\mat W$-linear representation of $\vec x$-nonlinear functions $\vec \phi(\vec x)$ in $\cal H^M$. Note that there exists a one-to-one correspondence between $\vec \phi$ and $\mat W$: for every $\vec \phi: \R^N \to \R^{M}$, there exists a $\mat W \in \R^{M \times L}$, and vice versa. 

\subsection{Neural Networks}\label{subsec:pre-NN}
Neural networks (NN) are another powerful tool to represent (i.e., approximate) nonlinear functions. A neural network function space (NNFS) $\cal K$ characterizes (or parameterizes) a set of multi-input multi-output functions. Typical choices are multi-layer feed-forward neural networks, recurrent neural networks, etc. For combining and estimation of wireless signals, the multi-layer feed-forward neural networks are standard \cite{ye2017power,he2020model,van2022transfer}. Suppose that we have $R - 1$ hidden layers (so in total $R+1$ layers including one input layer and one output layer) and each layer $r = 0, 1, \ldots, R$ contains $T_r$ neurons. To represent a function $\vec \phi: \R^N \to \R^M$, for the input layer $r = 0$ and output layer $r = R$, we have $T_0 = N$ and $T_{R} = M$, respectively. Let the output of the $r^\th$ layer be $\vec y_r \in \R^{T_r}$. For every layer $r$, we have $\vec y_r = \vec \sigma_r (\mat W^\circ_r \cdot \vec y_{r-1} + \vec b_r)$ where $\mat W^\circ_r \in \R^{T_{r} \times T_{r-1}}$ is the weight matrix, $\vec b_r \in \R^{T_r}$ is the bias vector, and the multi-output function $\vec \sigma_r$ is the activation function which is entry-wise identical. Hence, every function $\vec \phi: \R^N \to \R^M$ in a NNFS can be recursively expressed as \cite[Chap.~5]{bishop2006pattern}, \cite{li2023towards}
\begin{equation}\label{eq:non-lin-func-NNFS}
\begin{array}{cll}
\vec \phi(\vec x) &= \vec \sigma_{R} (\mat W_{R} \cdot [\vec y_{R-1}(\vec x);~1]) \\

\vec y_r(\vec x) &= \vec \sigma_r (\mat W_r \cdot [\vec y_{r-1}(\vec x);~1]), & r \in [R-1]\\

\vec y_0(\vec x) &= \vec x,
\end{array}
\end{equation}
where $\mat W_r \defeq [\mat W^\circ_r,~\vec b_r]$ for $r \in [R]$. Note that the activation functions can vary from one layer to another. 

\section{Problem Formulation}\label{sec:formulation}
Consider a narrow-band wireless signal transmission model
\begin{equation}\label{eq:signal-model}
    \rvec x = \mat H \rvec s + \rvec v
\end{equation}
where $\rvec x \in \C^{N}$ is the received signal, $\rvec s \in \C^{M}$ is the transmitted signal, $\mat H  \in \C^{N \times M}$ is the channel matrix, and $\rvec v \in \C^{N}$ is the zero-mean channel noise. The precoding operation (if exists) is integrated in $\mat H$. 
The transmitted symbols $\rvec s$ have zero means, which may be not only discrete symbols from constellations such as quadrature amplitude modulation but also arbitrary values such as integrated sensing and communication signals. We consider $L$ pilots $\rmat S \defeq (\rvec s_1,\rvec s_2,\ldots,\rvec s_L)$ in each frame, and the corresponding received symbols are $\rmat X \defeq (\rvec x_1,\rvec x_2,\ldots,\rvec x_L)$ under the noise $(\rvec v_1,\rvec v_2,\ldots,\rvec v_L)$.  
We suppose that $\mat R_s \defeq \E\rvec s \rvec s^\H$ and $\mat R_v \defeq \E\rvec v \rvec v^\H$ may not be identity or diagonal matrices: i.e., the components of $\rvec s$ can be correlated (e.g., in ISAC), so can be these of $\rvec v$. Consider the real-space representation of the signal model \eqref{eq:signal-model} by stacking the real and imaginary components: 
\begin{equation}\label{eq:real-signal-model}
\rvecul x = \ul{\matul H} \cdot \rvecul s + \rvecul v,
\end{equation}
where $\rvecul x \in \R^{2N}$, $\ul{\matul H} \in \R^{2N \times 2M}$, $\rvecul s \in \R^{2M}$, and $\rvecul v \in \R^{2N}$. 
The expressions of $\mat R_{\ul x} \defeq \E{\rvecul x \rvecul x^\T}$, $\mat R_{\ul s} \defeq \E{\rvecul s \rvecul s^\T}$, $\mat R_{\ul x \ul s} \defeq \E{\rvecul x \rvecul s^\T}$, and $\mat R_{\ul v} \defeq \E{\rvecul v \rvecul v^\T}$ can be readily obtained; see Appendix \ref{append:details-real-representation}. In some cases, signal estimation in real spaces can be technically simpler than that in complex spaces.

\subsection{Optimal Estimation}\label{subsec:nonlin-estimation}
\subsubsection{Optimal Nonlinear Estimation (Receive Combining)}
To recover $\rvec s$ using $\rvec x$, we consider an estimator $\rvech s \defeq \vec \phi(\rvec x)$, called a receive combiner, at the receiver where $\vec \phi: \C^{N} \to \C^{M}$ is a Borel-measurable function. Note that $\vec \phi(\rvec x)$ may be nonlinear in general because the joint distribution of $(\rvec x, \rvec s)$ is not necessarily Gaussian, for example, when the channel noise $\rvec v$ is non-Gaussian or when the power amplifiers work in non-linear regions. The signal estimation problem at the receiver can be written as a statistical machine-learning problem under the joint data distribution $\P_{\rvec x, \rvec s}$ of $(\rvec x, \rvec s)$, that is,
\begin{equation}\label{eq:opt-esti}
    \min_{\vec \phi \in \cal B_{\C^{N} \to \C^{M}}} \Tr \E_{\rvec x, \rvec s}[\vec \phi(\rvec x) - \rvec s][\vec \phi(\rvec x) - \rvec s]^\H,
\end{equation}
where $\cal B_{\C^{N} \to \C^{M}}$ contains all Borel-measurable estimators from $\C^{N}$ to $\C^{M}$. In what follows, we omit the notational dependence on $\C^{N}$ and $\C^{M}$, and use $\cal B$ as a shorthand. The optimal estimator, in the sense of minimum mean-squared error, is known as the conditional mean of $\rvec s$ given $\rvec x$, i.e.,
\begin{equation}\label{eq:nonlin-BF}
    \rvech s = \vec \phi (\rvec x) = \E({\rvec s | \rvec x}).
\end{equation}
Usually, it is computationally complicated to find the optimal $\vec \phi(\cdot)$ from the whole space $\cal B$ of Borel-measurable functions, that is, to compute the conditional mean. Therefore, in practice, we may find the optimal approximation of $\vec \phi(\cdot)$ in an RKHS $\cal H$ or a NNFS $\cal K$; note that $\cal H$ and $\cal K$ are two subspaces of $\cal B$. However, both $\cal H$ and $\cal K$ are sufficiently rich because they can be dense in the space of all continuous bounded functions.

\subsubsection{Optimal Linear Estimation (Receive Beamforming)}\label{subsubsec:beamforming}
If $\rvec x$ and $\rvec s$ are jointly Gaussian (e.g., when $\rvec s$ and $\rvec v$ are jointly Gaussian), the optimal estimator $\vec \phi$ is linear in $\rvec x$:
\begin{equation}\label{eq:BF}
    \rvech s = \mat W \rvec x,
\end{equation}
where $\mat W \in \C^{M \times N}$ is called a receive beamformer or a linear receive combiner. In this linear case, \eqref{eq:opt-esti} reduces to the usual Wiener--Hopf beamforming problem
\begin{equation}\label{eq:opt-BF}
    \min_{\mat W} \Tr \E_{\rvec x, \rvec s}[\mat W \rvec x - \rvec s][\mat W \rvec x - \rvec s]^\H,
\end{equation}
that is,
\begin{equation}\label{eq:opt-BF-explicit}
\min_{\mat W} \Tr \big[\mat W \mat R_{x} \mat W^\H - \mat W \mat R_{xs} - \mat R^\H_{xs} \mat W^\H + \mat R_{s}\big],
\end{equation}
where $\mat R_{x} \defeq \E{\rvec x \rvec x^\H} \in \C^{N \times N}$ and $\mat R_{xs} \defeq \E{\rvec x \rvec s^\H} \in \C^{N \times M}$. Since $\mat R_x = \mat H \mat R_s \mat H^\H + \mat R_v$ and $\mat R_{xs} = \mat H \mat R_s + \E{\rvec v \rvec s^\H} = \mat H \mat R_s$, the solution of \eqref{eq:opt-BF-explicit}, or \eqref{eq:opt-BF}, is 
\begin{equation}\label{eq:BF-Wiener}
\begin{array}{cl}
    \mat W^\star_{\text{Wiener}} &= \mat R^{\H}_{xs}\mat R^{-1}_{x} \\
    &= \mat R_s \mat H^\H [\mat H \mat R_s \mat H^\H + \mat R_v]^{-1},
\end{array}
\end{equation}
which is known as the Wiener beamformer. With an additional constraint $\mat W \mat H = \mat I_M$ (i.e., distortionless response), \eqref{eq:opt-BF-explicit} gives the Capon beamformer. Both the Wiener beamformer and the Capon beamformer maximize the output signal--to--interference-plus-noise ratio (SINR); hence, both are optimal in the sense of maximum output SINR. 

No matter whether $\P_{\rvec x, \rvec s}$ is Gaussian or not, \eqref{eq:opt-BF} or \eqref{eq:opt-BF-explicit} identifies the \bfit{optimal linear estimator} in the sense of minimum mean-squared
error among all linear estimators.

\subsubsection{Role of Channel Estimation}\label{subsubsec:role-CE} Eqs. \eqref{eq:opt-esti} and \eqref{eq:opt-BF} imply that channel estimation is not a necessary step in receive combining. The only necessary element, from the perspective of statistical machine learning, is the joint distribution $\P_{\rvec x, \rvec s}$ of the received signal $\rvec x$ and the transmitted signal $\rvec s$. Therefore, the following two points can be highlighted.
\begin{enumerate}[a)]
    \item If the joint distribution $\P_{\rvec x, \rvec s}$ is non-Gaussian, we just need to learn the mapping $\vec \phi$ using \eqref{eq:opt-esti}.
    
    \item If the joint distribution $\P_{\rvec x, \rvec s}$ is (or assumed to be) Gaussian, we just learn covariance matrices $\mat R_{xs}$ and $\mat R_x$; cf. \eqref{eq:BF-Wiener}; Gaussianity assumption of $\P_{\rvec x, \rvec s}$ is beneficial in reducing computational burdens. If, further, the channel matrix $\mat H$ is known, $\mat R_{xs}$ and $\mat R_x$ can be expressed using $\mat H$.
\end{enumerate}

\subsection{Distributional Uncertainty and Distributional Robustness}\label{subsec:dist-uncertainty}
For ease of conceptual illustration, we start with the following stationary-channel assumption in this subsection: The channel statistics remain unchanged within the communication frame so that the joint distribution $\P_{\rvec x, \rvec s}$ is fixed over time. That is, pilot data $\{(\vec x_1, \vec s_1), (\vec x_2, \vec s_2), \ldots, (\vec x_L, \vec s_L)\}$ and non-pilot communication data are drawn from the same unknown distribution $\P_{\rvec x, \rvec s}$. For the general case where the channel is not statistically stationary within a frame, see Appendix \ref{append:dist-uncertainty}; the statistical non-stationarity of $\P_{\rvec x, \rvec s}$ may be due to the time-selectivity of the transmit power matrix $\mat R_s$, of the channel matrix $\mat H$, and/or of the channel noise covariance $\mat R_v$.

\subsubsection{Issue of Distributional Uncertainty}
In practice, the true joint distribution $\P_{\rvec x, \rvec s}$ is unknown but can be estimated by the pilot data. Hence, the estimation of wireless signals is a data-driven statistical inference (i.e., statistical machine learning) problem. We let 
\begin{equation}\label{eq:empirical-dist}
\Ph_{\rvec x, \rvec s} \defeq \frac{1}{L} \sum^L_{i=1} \delta_{(\vec x_i, \vec s_i)}  
\end{equation}
denote the empirical distribution supported on the $L$ collected data $\{(\vec x_i, \vec s_i)\}_{i \in [L]}$, where $\delta_{(\vec x_i, \vec s_i)}$ denotes the Dirac distribution (i.e., point-mass distribution) centered on $(\vec x_i, \vec s_i)$; note that $\Ph_{\rvec x, \rvec s}$ is a discrete distribution. If we use the estimated joint distribution $\Ph_{\rvec x, \rvec s}$ as a surrogate of the true joint distribution $\P_{\rvec x, \rvec s}$, \eqref{eq:opt-esti} becomes the conventional empirical risk minimization (ERM) 
\begin{equation}\label{eq:esti-erm-dist}
\min_{\vec \phi \in \cal B} \Tr \E_{(\rvec x, \rvec s) \sim \Ph_{\rvec x, \rvec s}}[\vec \phi(\rvec x) - \rvec s][\vec \phi(\rvec x) - \rvec s]^\H,
\end{equation} 
i.e.,
\begin{equation}\label{eq:esti-erm}
\min_{\vec \phi \in \cal B}  \Tr \frac{1}{L} \sum^L_{i=1}[\vec \phi(\vec x_i) - \vec s_i][\vec \phi(\vec x_i) - \vec s_i]^\H.    
\end{equation}
Likewise, \eqref{eq:opt-BF-explicit} become the conventional beamforming problem
\begin{equation}\label{eq:BF-erm}
\displaystyle \min_{\mat W} \Tr \big[\mat W \math R_{x} \mat W^\H - \mat W \math R_{xs} - \math R^\H_{xs} \mat W^\H + \math R_{s}\big],
\end{equation}
where ${\math R}_x$, ${\math R}_{xs}$, and ${\math R}_s$ are the training-sample-estimated (i.e., nominal) values of $\mat R_{x}$, $\mat R_{xs}$, and $\mat R_{s}$, respectively.

There exists the distributional difference between the sample-defined nominal distribution $\Ph_{\rvec x, \rvec s}$ and true data-generating distribution $\P_{\rvec x, \rvec s}$ due to the limited size of the training data set (i.e., limited pilot length) and the time-selectivity of $\P_{\rvec x, \rvec s}$. From the perspective of applied statistics and machine learning, the distributional difference between $\Ph_{\rvec x, \rvec s}$ and $\P_{\rvec x, \rvec s}$ (i.e., the distributional uncertainty of $\Ph_{\rvec x, \rvec s}$ compared to $\P_{\rvec x, \rvec s}$) may cause significant performance degradation of \eqref{eq:esti-erm} compared to \eqref{eq:opt-esti}, so is the performance deterioration of \eqref{eq:BF-erm} compared to \eqref{eq:opt-BF-explicit}. For extensive reading on this point, see Appendix \ref{append:dist-uncertainty}. Therefore, to reduce the adverse effect introduced by the distributional uncertainty in $\Ph_{\rvec x, \rvec s}$, a new surrogate of \eqref{eq:opt-esti} rather than the sample-averaged approximation in \eqref{eq:esti-erm} is expected.

\subsubsection{Distributionally Robust Estimation}\label{subsec:dist-robust-BF}
To combat the distributional uncertainty in $\Ph_{\rvec x, \rvec s}$, we consider the distributionally robust counterpart of \eqref{eq:opt-esti}
\begin{equation}\label{eq:opt-esti-dist-robust}
    \min_{\vec \phi \in \cal B} \max_{\P_{\rvec x, \rvec s} \in \cal U_{\rvec x, \rvec s}} \Tr \E_{\rvec x, \rvec s}[\vec \phi(\rvec x) - \rvec s][\vec \phi(\rvec x) - \rvec s]^\H,
\end{equation}
where $\cal U_{\rvec x, \rvec s}$, called a distributional uncertainty set, contains a collection of distributions that are close to the nominal distribution (i.e., the sample-estimated distribution) $\Ph_{\rvec x, \rvec s}$;
\begin{equation}\label{eq:wasserstein-ball}
    \cal U_{\rvec x, \rvec s} \defeq \{\P_{\rvec x, \rvec s} |~ d(\P_{\rvec x, \rvec s}, \Ph_{\rvec x, \rvec s}) \le \epsilon \},
\end{equation}
where $d(\cdot, \cdot)$ denotes a similarity measure (e.g., metric or divergence) between two distributions and $\epsilon \ge 0$ an uncertainty quantification level. Since $\Ph_{\rvec x, \rvec s}$ is discrete and $\P_{\rvec x, \rvec s}$ is not, the Wasserstein distance \cite[Def.~2]{shafieezadeh2019regularization} and the maximum mean discrepancy (MMD) distance \cite[Def.~2.1]{staib2019distributionally} are the typical choices of $d(\cdot, \cdot)$ to construct $\cal U_{\rvec x, \rvec s}$. When $\Ph_{\rvec x, \rvec s}$ and $\P_{\rvec x, \rvec s}$ are parametric distributions (e.g., Gaussian, exponential family), divergences such as the Kullback-–Leibler (KL) divergence, or more general $\phi$-divergence, are also applicable to particularize $d(\cdot, \cdot)$ because parameters can be estimated using samples. 
When $\epsilon = 0$, \eqref{eq:opt-esti-dist-robust} reduces to \eqref{eq:esti-erm}. 

If $\cal U_{\rvec x, \rvec s}$ contains (or is assumed, for computational simplicity, to contain) only Gaussian distributions, \eqref{eq:opt-esti-dist-robust} particularizes to
\begin{equation}\label{eq:BF-dist-robust}
    \begin{array}{cl}
        \displaystyle \min_{\mat W} \max_{\mat R} &\Tr \big[\mat W \mat R_{x} \mat W^\H - \mat W \mat R_{xs} - \mat R^\H_{xs} \mat W^\H + \mat R_{s}\big] \\
        \st & d_0(\mat R,~\math R) \le \epsilon_0, \\

        & \mat R \succeq \mat 0,
    \end{array}
\end{equation}
where 
\begin{equation}\label{eq:big-R}
    \mat R \defeq 
    \left[
        \begin{array}{cc}
           \mat R_x  &  \mat R_{xs} \\
           \mat R^\H_{xs}  &  \mat R_s
        \end{array}
    \right],
    ~~~
    \math R \defeq
    \left[
        \begin{array}{cc}
           \math R_x  &  \math R_{xs} \\
           \math R^\H_{xs}  &  \math R_s
        \end{array}
    \right],
\end{equation}
because every zero-mean complex Gaussian distribution is uniquely characterized by its covariance and pseudo-covariance, but in receive beamforming, we do not consider pseudo-covariances; cf. \eqref{eq:BF-Wiener}; $d_0$ denotes the matrix similarity measures (e.g., matrix distances); $\epsilon_0 \ge 0$ is the uncertainty quantification parameter. When $\epsilon_0 = 0$, \eqref{eq:BF-dist-robust} reduces to \eqref{eq:BF-erm}.

For additional discussions on the framework of distributionally robust estimation, see Appendix \ref{append:DRO-supplementary}.

\section{Distributionally Robust Linear Estimation}\label{sec:lin-esti}

Due to several practical benefits of linear estimation, for example, the simplicity of hardware structures, the clarity of physical meaning (i.e., constructive and destructive interference through beamforming), and the easiness of computations, investigating distributionally robust linear estimation problems is important. This section particularly studies Problem \eqref{eq:BF-dist-robust}.

\subsection{General Framework and Concrete Examples}
The following lemma solves Problem \eqref{eq:BF-dist-robust}.
\begin{lemma}\label{lem:lin-BF-dist-robust-dual}
Suppose that the set $\{\mat R|~d_0(\mat R,~\math R) \le \epsilon_0\}$ is compact convex and $\mat R_x$ is invertible. Let $\mat R^\star$ solve the problem below:
\begin{equation}\label{eq:BF-dist-robust-max}
    \begin{array}{cl}
        \displaystyle \max_{\mat R} &\Tr \big[ -\mat R^{\H}_{xs}\mat R^{-1}_{x} \mat R_{xs} + \mat R_{s}\big] \\
        \st & d_0(\mat R,~\math R) \le \epsilon_0, \\
        & \mat R \succeq \mat 0,~~~\mat R_x \succ \mat 0.
    \end{array}
\end{equation}
Construct $\mat W^\star$ using $\mat R^\star$ as follows: 
\begin{equation}\label{eq:dist-robust-BF}
    \mat W^\star \defeq \mat R^{\star \H}_{xs}\mat R^{\star -1}_{x}.
\end{equation}
Then $(\mat W^\star, \mat R^\star)$ is a solution to Problem \eqref{eq:BF-dist-robust}. On the other hand, if $(\mat W^\star, \mat R^\star)$ solves Problem \eqref{eq:BF-dist-robust}, then $\mat R^\star$ is a solution to \eqref{eq:BF-dist-robust-max} and $(\mat W^\star, \mat R^\star)$ satisfies \eqref{eq:dist-robust-BF}.
\end{lemma}
\begin{proof}
See Appendix \ref{append:dual}. \stp
\end{proof}

Let 
\begin{equation}\label{eq:BF-obj-f1}
    f_1(\mat R) \defeq \Tr \big[ -\mat R^{\H}_{xs}\mat R^{-1}_{x} \mat R_{xs} + \mat R_{s}\big]
\end{equation}
denote the objective function of \eqref{eq:BF-dist-robust-max}. When $\mat R_s$ and $\mat R_{xs}$ are fixed, we define 
\begin{equation}\label{eq:BF-obj-f2}
    f_2(\mat R_x) \defeq \Tr \big[ -\mat R^{\H}_{xs}\mat R^{-1}_{x} \mat R_{xs} + \mat R_{s}\big].
\end{equation}

The theorem below studies the properties of $f_1$ and $f_2$.
\begin{theorem}\label{thm:f-increasing-R-x}
Consider the definition of $\mat R$ in \eqref{eq:big-R}. The functions $f_1$ defined in \eqref{eq:BF-obj-f1} and $f_2$ defined in \eqref{eq:BF-obj-f2} are monotonically increasing in $\mat R$ and $\mat R_x$, respectively. To be specific, if $\mat R_1 \succeq \mat R_2 \succeq \mat 0$, $\mat R_{1, x} \succ \mat 0$, and $\mat R_{2, x} \succ \mat 0$, we have $f_1(\mat R_1) \ge f_1(\mat R_2)$. In addition, if $\mat R_{1, x} \succeq \mat R_{2, x} \succ \mat 0$, we have $f_2(\mat R_{1, x}) \ge f_2(\mat R_{2, x})$.
\end{theorem}
\begin{proof}
    See Appendix \ref{append:f-increasing-R-x}. \stp
\end{proof}

To concretely solve \eqref{eq:BF-dist-robust-max}, we need to particularize $d_0$. This article investigates the following uncertainty sets.

\begin{definition}[Additive Moment Uncertainty Set]\label{def:uncertainty-set-additive-moment}
The additive moment uncertainty set of $\mat R$ is constructed as
\begin{equation}\label{eq:uncertainty-set-additive-moment}
\{\mat R|~\math R - \epsilon_0 \mat E \preceq \mat R \preceq \math R + \epsilon_0 \mat E,~\mat R \succeq \mat 0\}
\end{equation}
for some $\mat E \succeq \mat 0$ and $\epsilon_0 \ge 0$. \stp
\end{definition}

Definition \ref{def:uncertainty-set-additive-moment} is motivated by the fact that the difference $\mat R - \math R$ is bounded by some threshold matrix $\mat E$ and error quantification level $\epsilon_0$: specifically, $- \epsilon_0 \mat E \preceq \mat R - \math R \preceq \epsilon_0 \mat E$. In practice, we can consider the threshold as an identity matrix because, for every non-identity $\mat E \succeq \mat 0$, we have $\mat E \preceq \lambda_1 \mat I_{N + M}$ where $\lambda_1$ is the largest eigenvalue of $\mat E$.

\begin{definition}[Diagonal-Loading Uncertainty Set]\label{def:uncertainty-set-diag-loading}
The diagonal-loading uncertainty set of $\mat R$ is constructed as
\begin{equation}\label{eq:uncertainty-set-diag-loading}
\{\mat R|~\math R - \epsilon_0 \mat I_{N + M} \preceq \mat R \preceq \math R + \epsilon_0 \mat I_{N + M},~\mat R \succeq \mat 0\}
\end{equation}
for some $\epsilon_0 \ge 0$. \stp
\end{definition}

Due to the concentration property of the sample-covariance $\math R$ to the true covariance $\mat R$ when the true distribution $\P_{\rvec x, \rvec s}$ is fixed within a frame, finite values of $\epsilon_0$ exist for every sample size $L$;  
NB: $\epsilon_0 \to 0$ as $L \to \infty$. However, given $L$, the smallest $\epsilon_0$ cannot be practically calculated because it depends on the true but unknown $\P_{\rvec x, \rvec s}$. If $\mat E$ is block-diagonal, the generalized diagonal-loading uncertainty set can be motivated.

\begin{definition}[Generalized Diagonal-Loading Uncertainty Set]\label{def:generalized-uncertainty-set-diag-loading}
The generalized diagonal-loading uncertainty set of $\mat R$ is constructed by the following constraints: $\mat R \succeq \mat 0$ and 
\begin{equation}\label{eq:generalized-uncertainty-set-diag-loading}
\begin{array}{l}
        \left[
        \begin{array}{cc}
           \math R_x  &  \math R_{xs} \\
           \math R^\H_{xs}  &  \math R_s
        \end{array}
        \right] - \epsilon_0 
        \left[
        \begin{array}{cc}
           \mat F  &  \mat 0 \\
           \mat 0  &  \mat G
        \end{array}
        \right] \\
        
        \quad \quad \preceq
        \left[
        \begin{array}{cc}
           \mat R_x  &  \mat R_{xs} \\
           \mat R^\H_{xs}  &  \mat R_s
        \end{array}
        \right] \\
        
        \quad \quad \quad \quad \preceq
        \left[
        \begin{array}{cc}
           \math R_x  &  \math R_{xs} \\
           \math R^\H_{xs}  &  \math R_s
        \end{array}
        \right] + \epsilon_0 
        \left[
        \begin{array}{cc}
           \mat F  &  \mat 0 \\
           \mat 0  &  \mat G
        \end{array}
        \right],
\end{array}
\end{equation}
for some $\mat F,\mat G \succeq \mat 0$ and $\epsilon_0 \ge 0$.
\stp
\end{definition}

Definitions \ref{def:uncertainty-set-additive-moment}, \ref{def:uncertainty-set-diag-loading}, and \ref{def:generalized-uncertainty-set-diag-loading} are introduced for the first time in this article. Another type of moment-based uncertainty set is popular in the literature, which we refer to as the multiplicative moment uncertainty set for differentiation.
\begin{definition}[Multiplicative Moment Uncertainty Set \cite{delage2010distributionally}]\label{def:uncertainty-set-multiplicative-moment}
The multiplicative moment uncertainty set of $\mat R$ is given as
\begin{equation}\label{eq:uncertainty-set-multiplicative-moment}
\{\mat R|~\theta_1 \math R \preceq \mat R \preceq \theta_2 \math R\}
\end{equation}
for some $\theta_2 \ge 1 \ge \theta_1 \ge 0$. \stp
\end{definition}

The following corollary shows the distributionally robust linear beamformers associated with the various uncertainty sets in Definitions \ref{def:uncertainty-set-additive-moment}, \ref{def:uncertainty-set-diag-loading}, \ref{def:generalized-uncertainty-set-diag-loading}, and \ref{def:uncertainty-set-multiplicative-moment}.

\begin{corollary}[of Theorem \ref{thm:f-increasing-R-x}]\label{cor:solutions-under-moments}
Consider the moment-based uncertainty sets in Definitions \ref{def:uncertainty-set-additive-moment}, \ref{def:uncertainty-set-diag-loading}, \ref{def:generalized-uncertainty-set-diag-loading}, and \ref{def:uncertainty-set-multiplicative-moment}. The distributionally robust linear beamforming \eqref{eq:BF-dist-robust-max} is analytically solved by the corresponding upper bounds of $\mat R$. To be specific,
\begin{enumerate}[\quad C1)]
    \item Under Definition \ref{def:uncertainty-set-additive-moment}, the additive-moment distributionally robust (DR-AM) beamformer is 
    \begin{equation}\label{eq:DRBF-AM}
    \begin{array}{cl}
        \mat W^\star_{\text{DR-AM}} &= (\math R_{xs} + \epsilon_0 \mat E_{xs})^{\H} (\math R_{x} + \epsilon_0 \mat E_{x})^{-1} \\
        &= (\math H \math R_{s} + \epsilon_0 \mat E_{xs})^{\H} \cdot \\
        & \quad \quad \quad [\math H \math R_s \math H^\H  + \math R_v + \epsilon_0 \mat E_x]^{-1},
    \end{array}
    \end{equation}
    where $\math H$, $\math R_{s}$, and $\math R_{v}$ denote the estimates of $\mat H$, $\mat R_{s}$, and $\mat R_{v}$, respectively.

    \item Under Definition \ref{def:uncertainty-set-diag-loading}, the diagonal-loading distributionally robust (DR-DL) beamformer is 
    \begin{equation}\label{eq:DRBF-DL}
    \begin{array}{cl}
        \mat W^\star_{\text{DR-DL}} &= \math R^{\H}_{xs} [{\math R}_x + \epsilon_0 \mat I_N]^{-1} \\
        &= \math R_s \math H^\H [\math H \math R_s \math H^\H  + \math R_v + \epsilon_0 \mat I_N]^{-1},
    \end{array}
    \end{equation}
    which is also known as the loaded sample matrix inversion method \cite{cox1987robust}, \cite[Eq.~(11)]{lorenz2005robust} and widely-used in the practice of wireless communications.

    \item Under Definition \ref{def:generalized-uncertainty-set-diag-loading}, the generalized diagonal-loading distributionally robust beamformer (DR-GDL) is 
    \begin{equation}\label{eq:DRBF-GDL}
    \begin{array}{cl}
        \mat W^\star_{\text{DR-GDL}} &= \math R^{\H}_{xs} [{\math R}_x + \epsilon_0 \mat F]^{-1} \\
        &= \math R_s \math H^\H [\math H \math R_s \math H^\H  + \math R_v + \epsilon_0 \mat F]^{-1}.
    \end{array}
    \end{equation}

    \item Under Definition \ref{def:uncertainty-set-multiplicative-moment}, the multiplicative-moment (MM) distributionally robust beamformer is identical to the Wiener beamformer \eqref{eq:BF-Wiener} at nominal values:
    \begin{equation}\label{eq:DRBF-MM}
    \begin{array}{cl}
        \mat W^\star_{\text{DR-MM}} &= \math R_{xs}^{\H} \math R_{x}^{-1} \\
        &= \math R_{s} \math H^{\H} [\math H \math R_s \math H^\H  + \math R_v]^{-1}.
    \end{array}
    \end{equation}
\end{enumerate}
The corresponding estimation errors are simple to obtain.
\stp
\end{corollary}

Corollary \ref{cor:solutions-under-moments} implies that, in the sense of the same induced robust beamformers, the diagonal-loading uncertainty set \eqref{eq:uncertainty-set-diag-loading} and the generalized diagonal-loading uncertainty set \eqref{eq:generalized-uncertainty-set-diag-loading} are technically equivalent to the following trimmed versions.
\begin{definition}[Trimmed Diagonal-Loading Uncertainty Sets]\label{def:trimmed-sets}
By setting $\mat G \defeq \mat 0$ in \eqref{eq:generalized-uncertainty-set-diag-loading}, in terms of $\mat R_x$, \eqref{eq:generalized-uncertainty-set-diag-loading} reduces to the trimmed generalized diagonal-loading uncertainty set:
\begin{equation}\label{eq:trimmed-generalized-uncertainty-set-diag-loading}
\{\mat R_x|~\math R_x - \epsilon_0 \mat F \preceq \mat R_x \preceq \math R_x + \epsilon_0 \mat F,~\mat R_x \succeq \mat 0\}.
\end{equation}
The trimmed diagonal-loading uncertainty set
\begin{equation}\label{eq:trimmed-uncertainty-set-diag-loading}
\{\mat R_x|~\math R_x - \epsilon_0 \mat I_N \preceq \mat R_x \preceq \math R_x + \epsilon_0 \mat I_N,~\mat R_x \succeq \mat 0\},
\end{equation}
is obtained by letting $\mat F \defeq \mat I_N$. \stp
\end{definition}

The robust beamformers corresponding to the trimmed uncertainty sets \eqref{eq:trimmed-generalized-uncertainty-set-diag-loading} and \eqref{eq:trimmed-uncertainty-set-diag-loading} remain the same as defined in \eqref{eq:DRBF-GDL} and \eqref{eq:DRBF-DL}, respectively; cf. Theorem \ref{thm:f-increasing-R-x}.

As we can see from Corollary \ref{cor:solutions-under-moments}, the primary benefit of using the moment-based uncertainty sets is the computational simplicity due to the availability of closed-form solutions. If the uncertainty sets are constructed using the Wasserstein distance
$
\sqrt{\Tr[{\mat R + \math R - 2 (\math R^{1/2} \mat R \math R^{1/2})^{1/2}}]} \le \epsilon_0
$ 
or the KL divergence
$
\frac{1}{2} [\Tr[{\math R^{-1} \mat R - \bm I_{N+M}}] - \ln \det (\math R^{-1} \mat R) ] \le \epsilon_0
$ 
between $\cal{CN}(\mat 0, \mat R)$ and $\cal{CN}(\mat 0, \math R)$, the induced distributionally robust linear beamforming problems have no closed-form solutions, and therefore, are computationally prohibitive in practice. In addition, Corollary \ref{cor:solutions-under-moments} suggests that the distributionally robust beamformer under the multiplicative moment uncertainty set \eqref{eq:uncertainty-set-multiplicative-moment} is the same as the nominal beamformer $\math R^{\H}_{xs} \math R_{x}^{-1}$, which essentially does not introduce robustness in wireless signal estimation; this is another motivation why we construct new moment-based uncertainty sets in Definitions \ref{def:uncertainty-set-additive-moment}, \ref{def:uncertainty-set-diag-loading}, and \ref{def:generalized-uncertainty-set-diag-loading}. However, we can modify the multiplicative moment uncertainty set in Definition \ref{def:uncertainty-set-multiplicative-moment} to achieve robustness.
\begin{definition}[Modified Multiplicative Moment Uncertainty Set]\label{def:uncertainty-set-multiplicative-moment-modified}
The modified multiplicative moment uncertainty set of $\mat R$ is defined by the following constraint:
\begin{equation}\label{eq:uncertainty-set-multiplicative-moment-modified}
        \left[
        \begin{array}{cc}
           \theta_1 \math R_x  &  \math R_{xs} \\
           \math R^\H_{xs}  &  \theta_1 \math R_s
        \end{array}
        \right] 
        \preceq
        \left[
        \begin{array}{cc}
           \mat R_x  &  \mat R_{xs} \\
           \mat R^\H_{xs}  &  \mat R_s
        \end{array}
        \right] 
        \preceq
        \left[
        \begin{array}{cc}
           \theta_2 \math R_x  &  \math R_{xs} \\
           \math R^\H_{xs}  &  \theta_2 \math R_s
        \end{array}
        \right]
\end{equation}
for some $\theta_2 \ge 1 \ge \theta_1 \ge 0$ such that the left-most matrix is positive semi-definite. \stp
\end{definition}

The robust beamformer under the modified multiplicative moment uncertainty set \eqref{eq:uncertainty-set-multiplicative-moment-modified} is 
\begin{equation}
\mat W^\star_{\text{DR-MMM}} = \math R^\H_{xs} \cdot [\theta_2 \math R_x]^{-1}.
\end{equation}

In terms of the uncertainties of $\mat R_s$ and $\mat R_v$, Problem \eqref{eq:BF-dist-robust-max} can be explicitly written as
\begin{equation}\label{eq:BF-dist-robust-max-explicit}
    \begin{array}{cl}
        \displaystyle \max_{\mat R_s, \mat R_v} &\Tr \big[\mat R_s -\mat R_s \mat H^\H (\mat H \mat R_s \mat H^\H + \mat R_v)^{-1} \mat H \mat R_s \big] \\
        \st & d_1(\mat R_s, \math R_s) \le \epsilon_1, \\
        & d_2(\mat R_v, \math R_v) \le \epsilon_2, \\
        & \mat R_s \succeq \mat 0,~\mat R_v \succeq \mat 0,
    \end{array}
\end{equation}
for some similarity measures $d_1$ and $d_2$ and nonnegative scalars $\epsilon_1$ and $\epsilon_2$. For every given $(\mat R_s, \mat R_v)$, the associated beamformer is given in \eqref{eq:BF-Wiener}. When the uncertainty in the channel matrix must be investigated, we can consider
\begin{equation}\label{eq:BF-dist-robust-max-explicit-H}
    \begin{array}{cl}
        \displaystyle \max_{\mat H} &\Tr \big[\mat R_s -\mat R_s \mat H^\H (\mat H \mat R_s \mat H^\H + \mat R_v)^{-1} \mat H \mat R_s \big] \\
        \st & d_3(\mat H, \math H) \le \epsilon_3,
    \end{array}
\end{equation}
which is not a semi-definite program. In addition, the gradient of the objective function with respect to $\mat H$ is complicated to obtain. Hence, practically, we should avoid directly attacking Problem \eqref{eq:BF-dist-robust-max-explicit-H}; this can be done by directly considering the uncertainties of $\mat R_x$ and $\mat R_{xs}$ (i.e., $\mat R$) because the uncertainties of $\mat R_s$, $\mat R_v$, and $\mat H$ can be reflected in the uncertainties of $\mat R_x$ and $\mat R_{xs}$; cf. $\mat R_x = \mat H \mat R_s \mat H^\H + \mat R_v$ and $\mat R_{xs} = \mat H \mat R_s$.

In addition to Corollary \ref{cor:solutions-under-moments}, below we provide other concrete examples to further showcase the usefulness and applications of the distributionally robust beamforming formulations \eqref{eq:BF-dist-robust-max} and \eqref{eq:BF-dist-robust-max-explicit}, where the trimmed uncertainty sets are employed.

\begin{example}[Distributionally Robust Capon Beamforming]\label{exam:dist-robust-capon}
We consider a distributionally robust Capon beamforming problem under the trimmed uncertainty set \eqref{eq:trimmed-uncertainty-set-diag-loading}:
\[
    \begin{array}{cl}
        \displaystyle \min_{\mat W} \max_{\mat R_x} &\Tr \big[\mat W \mat R_{x} \mat W^\H - 2 \mat R_s + \mat R_{s}\big] \\
        \st & \mat W \mat H = \mat I_M, \\
        &  {\math R}_x - \epsilon_0 \mat I_N \preceq \mat R_x \preceq {\math R}_x + \epsilon_0 \mat I_N, \\
        & \mat R_x \succeq \mat 0,
    \end{array}
\]
which is equivalent, in the sense of the same solutions, to
\[
    \begin{array}{cl}
        \displaystyle \min_{\mat W} \max_{\mat R_x} &\Tr \big[\mat W \mat R_{x} \mat W^\H\big] \\
        \st & \mat W \mat H = \mat I_M, \\
        &  {\math R}_x - \epsilon_0 \mat I_N \preceq \mat R_x \preceq {\math R}_x + \epsilon_0 \mat I_N, \\
        & \mat R_x \succeq \mat 0.
    \end{array}
\]
According to Theorem \ref{thm:f-increasing-R-x}, the above display is equivalent to
\[
    \begin{array}{cl}
        \displaystyle \min_{\mat W} &\Tr \big[\mat W {\math R}_x  \mat W^\H\big] + \epsilon_0 \cdot \Tr \big[\mat W \mat W^\H\big] \\
        \st & \mat W \mat H = \mat I_M.
    \end{array}
\]
The above formulation is the squared-$F$-norm--regularized Capon beamformer \cite[Eq.~(10)]{lorenz2005robust} whose solution is
\begin{equation}\label{eq:diag-loading-capon}
\begin{array}{l}
\mat W^\star_{\text{DR-Capon}} = [\mat H^\H (\math R_x + \epsilon_0 \mat I_N)^{-1} \mat H]^{-1} \cdot \\
\quad \quad \quad \quad \quad \quad \quad \quad \quad \quad \quad \quad \mat H^\H (\math R_x + \epsilon_0 \mat I_N)^{-1},
\end{array}
\end{equation}
which is the diagonal-loading Capon beamformer.
\stp 
\end{example}

\begin{example}[Eigenvalue Thresholding]\label{exam:eigen-thres}
Suppose that $\math R_x$ admits the eigenvalues of $\operatorname{diag}\{\lambda_1, \lambda_2, \ldots, \lambda_N\}$ in descending order and the eigenvectors in $\mat Q$ (columns). Let $0 \le \mu \le 1$ be a shrinking coefficient. If we assume $\mat R_x \preceq \math R_{x,\text{thr}}$ where 
\begin{equation}\label{eq:eigen-thres}
\begin{array}{l}
\math R_{x,\text{thr}} \defeq \\
\quad \mat Q 
\left[ 
\begin{array}{cccc}
    \lambda_1 & & & \\
    & \max\{\mu \lambda_1, \lambda_2\}& & \\
    & & \ddots & \\
    & & & \max\{\mu \lambda_1, \lambda_N\}
\end{array}
\right] 
\mat Q^{-1},
\end{array}
\end{equation}
we have the distributionally robust beamformer
\begin{equation}\label{eq:eigen-thres-BF}
    \mat W^\star_{\text{DR-ET}} = \mat R^{\H}_{xs} \math R^{-1}_{x,\text{thr}},
\end{equation}
which is known as the eigenvalue thresholding method \cite{harmanci2000relationships}, \cite[Eq.~(12)]{lorenz2005robust}.
\stp
\end{example}

\begin{example}[Distributionally Robust Beamforming for Uncertain $\mat R_s$ and $\mat R_v$]\label{exam:uncertainty-M-R}
Consider Problem \eqref{eq:BF-dist-robust-max-explicit}. Since the objective of \eqref{eq:BF-dist-robust-max-explicit} is increasing in both $\mat R_s$ and $\mat R_v$,\footnote{This claim can be routinely proven in analogy to Theorem \ref{thm:f-increasing-R-x} and a real-space case in \cite[Thm.~1]{wang2023distributionally}.} if 
\[
\math R_s - \epsilon_1 \mat I_M \preceq \mat R_s \preceq \math R_s + \epsilon_1 \mat I_M,
\]
we have a distributionally robust beamformer
\begin{equation}\label{eq:diag-loading-using-H}
\begin{array}{cl}
    \mat W^\star_{\text{DR}} &= (\math R_s + \epsilon_1 \mat I_M) \mat H^\H [\mat H (\math R_s + \epsilon_1 \mat I_M) \mat H^\H  + \mat R_v]^{-1} \\
    &= (\math R_s + \epsilon_1 \mat I_M) \mat H^\H [\mat H \math R_s \mat H^\H + \mat R_v + \epsilon_1 \mat H \mat H^\H]^{-1};
\end{array}
\end{equation}
if instead
\begin{equation}\label{eq:eps-perturb-channel}
\math R_s - \epsilon_1 \mat H^\H (\mat H \mat H^\H)^{-2} \mat H \preceq \mat R_s \preceq \math R_s + \epsilon_1 \mat H^\H (\mat H \mat H^\H)^{-2} \mat H,
\end{equation}
we have
\begin{equation}\label{eq:diag-loading-M}
\begin{array}{l}
\mat W^\star_{\text{DR}} = [\math R_s \mat H^\H + \epsilon_1 \mat H^\H (\mat H \mat H^\H)^{-1}] \cdot \\ 
\quad \quad \quad \quad \quad [\mat H \math R_s \mat H^\H + \mat R_v + \epsilon_1 \mat I_N]^{-1},
\end{array}
\end{equation}
which is a modified diagonal-loading beamformer. 
On the other hand, if
\[
\math R_v - \epsilon_2 \mat I_N \preceq \mat R_v \preceq \math R_v + \epsilon_2 \mat I_N,
\]
we have
\begin{equation}\label{eq:diag-loading-R}
\mat W^\star_{\text{DR}} = \mat R_s \mat H^\H [\mat H \mat R_s \mat H^\H + \math R_v + \epsilon_2 \mat I_N]^{-1},
\end{equation}
which is also a diagonal-loading beamformer. \stp 
\end{example}

Motivated by Corollary \ref{cor:solutions-under-moments} and Examples \ref{exam:dist-robust-capon}$\sim$\ref{exam:uncertainty-M-R}, as well as the trimmed uncertainty sets in Definition \ref{def:trimmed-sets}, we have the following important theorem, which justifies the popular ridge regression in machine learning.

\begin{theorem}[Ridge Regression and Tikhonov Regularization]\label{thm:ridge-regression}
Consider a linear regression problem on $(\rvec x, \rvec s)$, i.e.,
\[
    \rvec s = \mat W \rvec x + \rvec e,
\]
where $\rvec e$ denotes the error term and the distributionally robust estimator of $\mat W$, i.e.,
\[
    \min_{\mat W \in \C^{M \times N}} \max_{\P_{\rvec x, \rvec s} \in \cal U_{\rvec x, \rvec s}} \Tr \E_{\rvec x, \rvec s}[\mat W \rvec x - \rvec s][\mat W \rvec x - \rvec s]^\H,
\]
which can be particularized to \eqref{eq:BF-dist-robust}. Supposing that the second-order moment of $\rvec x$ is uncertain and quantified as
\[
{\math R}_x - \epsilon_0 \mat I_N \preceq \mat R_x \preceq {\math R}_x + \epsilon_0 \mat I_N,
\]
then the distributionally robust estimator of $\mat W$ becomes a ridge regression (i.e., squared-$F$-norm regularized) method
\[
    \displaystyle \min_{\mat W} \Tr \big[\mat W \math R_{x} \mat W^\H - \mat W \math R_{xs} - \math R^\H_{xs} \mat W^\H + \math R_{s}\big] + \epsilon_0 \Tr \big[\mat W \mat W^\H\big].
\]
The regularization term becomes $\Tr \big[\mat W \mat F \mat W^\H\big]$, which is known as the Tikhonov regularizer, if 
\[
{\math R}_x - \epsilon_0 \mat F \preceq \mat R_x \preceq {\math R}_x + \epsilon_0 \mat F
\] 
for some $\mat F \succeq \mat 0$.
\end{theorem}
\begin{proof}
This is due to Lemma \ref{lem:lin-BF-dist-robust-dual} and Theorem \ref{thm:f-increasing-R-x}. Just note that 
$
\Tr \big[\mat W (\math R_{x} + \epsilon_0 \mat F) \mat W^\H - \mat W \math R_{xs} - \math R^\H_{xs} \mat W^\H + \math R_{s}\big] = \\
\Tr \big[\mat W \math R_{x} \mat W^\H - \mat W \math R_{xs} - \math R^\H_{xs} \mat W^\H + \math R_{s}\big] + \epsilon_0 \Tr \big[\mat W \mat F \mat W^\H\big]
$. This completes the proof.
\stp
\end{proof}

Note that in Theorem \ref{thm:ridge-regression}, the second-order moment of $\rvec s$ is not considered because it does not influence the optimal solution of $\mat W$: i.e., the optimal solution of $\mat W$ does not depend on the value of $\mat R_{s}$. Theorem \ref{thm:ridge-regression} gives a new theoretical interpretation of the popular ridge regression in machine learning from the perspective of distributional robustness against second-moment uncertainties of the feature vector $\rvec x$; another interpretation of ridge regression from the perspective of distributional robustness under martingale constraints is identified in \cite[Ex.~3.3]{li2022tikhonov}. When the uncertainty is quantified by the Wasserstein distance, a similar result can be seen in \cite[Prop.~3]{kuhn2019wasserstein}, \cite[Prop.~2]{blanchet2019robust}, which however is not a ridge regression formulation because in \cite[Prop.~3]{kuhn2019wasserstein} and \cite[Prop.~2]{blanchet2019robust}, the loss function is square-rooted and the norm regularizer is not squared; cf. also \cite[Rem.~18 and 19]{shafieezadeh2019regularization}.
The corollary below justifies the rationale of any norm-regularized method.

\begin{corollary}\label{cor:any-norm-regularization}
The following squared-norm-regularized beamforming formulation can combat the distributional uncertainty:
\begin{equation}\label{eq:ridge-any-norm}
    \displaystyle \min_{\mat W} \Tr \big[\mat W \math R_{x} \mat W^\H - \mat W \math R_{xs} - \math R^\H_{xs} \mat W^\H + \math R_{s}\big] + \lambda \|\mat W\|^2,
\end{equation}
where $\|\cdot\|$ denotes any matrix norm. This is because all norms on $\C^{M \times N}$ are equivalent; hence, there exists some $\lambda \ge 0$ such that $\lambda \|\mat W\|^2 \ge \epsilon_0 \|\mat W\|^2_F = \epsilon_0 \Tr \big[\mat W \mat W^\H\big]$. As a result, \eqref{eq:ridge-any-norm} can upper bound the ridge cost in Theorem \ref{thm:ridge-regression}.
\stp
\end{corollary}

Motivated by Theorem \ref{thm:ridge-regression}, the following corollary is immediate, which gives another interpretation of ridge regression and Tikhonov regularization from the perspective of data augmentation through data perturbation 
(cf. noise injection in image \cite{shorten2019survey} and speech \cite{saon2019sequence} processing).

\begin{corollary}[Data Augmentation for Linear Regression]\label{cor:adv-learning-linear}
Consider a linear regression problem on $(\rvec x, \rvec s)$ with data perturbation vectors $(\rvec \Delta_{x}, \rvec \Delta_{s})$
\[
    (\rvec s + \rvec \Delta_{s}) = \mat W (\rvec x + \rvec \Delta_{x}) + \rvec e,
\]
and the distributionally robust estimator of $\mat W$
\[
\begin{array}{l}
    \displaystyle \min_{\mat W \in \C^{M \times N}} \max_{\P_{\rvec \Delta_{x}, \rvec \Delta_{s}} \in \cal U_{{\rvec \Delta_{x}, \rvec \Delta_{s}}}} \Tr \E_{(\rvec x, \rvec s) \sim \Ph_{\rvec x, \rvec s}} \E_{\rvec \Delta_{x}, \rvec \Delta_{s}} \Big\{ \\
    
    \quad [\mat W (\rvec x + \rvec \Delta_{x}) - (\rvec s + \rvec \Delta_{s})][\mat W (\rvec x + \rvec \Delta_{x}) - (\rvec s + \rvec \Delta_{s})]^\H \Big\}.
\end{array}
\]
Suppose that $\rvec \Delta_{x}$ is uncorrelated with $\rvec x$, with $\rvec s$, and with $\rvec \Delta_{s}$; in addition, $\rvec \Delta_{s}$ is uncorrelated with $\rvec x$. If the second-order moment of $\rvec \Delta_{x}$ is upper bounded as 
$
\E \rvec \Delta_{x} \rvec \Delta^\H_{x} \preceq \epsilon_0 \mat I_N
$, 
then the distributionally robust estimator of $\mat W$ becomes a ridge regression (i.e., squared-$F$-norm regularized) method
\[
    \displaystyle \min_{\mat W} \Tr \big[\mat W \math R_{x} \mat W^\H - \mat W \math R_{xs} - \math R^\H_{xs} \mat W^\H + \math R_{s}\big] + \epsilon_0 \Tr \big[\mat W \mat W^\H\big].
\]
The regularization term becomes $\Tr \big[\mat W \mat F \mat W^\H\big]$, which is known as the Tikhonov regularizer, if
$
\E \rvec \Delta_{x} \rvec \Delta^\H_{x} \preceq \epsilon_0 \mat F
$, 
for some $\mat F \succeq \mat 0$.  \stp
\end{corollary}

The second-order moment of $\rvec \Delta_s$ is not considered in Corollary \ref{cor:adv-learning-linear} as it does not influence the optimal value of $\mat W$.

\subsection{Complex Uncertainty Sets}\label{subsec:BF-nontrivial-sets}
Below we remark on more general construction methods for the uncertainty set of $\mat R$ using the Wasserstein distance and the $F$-norm, beyond the moment-based methods in Definitions \ref{def:uncertainty-set-additive-moment}$\sim$\ref{def:uncertainty-set-multiplicative-moment-modified}. However, note that such complicated approaches are computationally prohibitive in practice when $N$ or $M$ is large.

\subsubsection{Wasserstein Distributionally Robust Beamforming}
We start with the Wasserstein distance:
\begin{equation}\label{eq:BF-dist-robust-max-wasserstein}
    \begin{array}{cl}
        \displaystyle \max_{\mat R} &\Tr \big[ -\mat R^{\H}_{xs}\mat R^{-1}_{x} \mat R_{xs} + \mat R_{s}\big] \\
        \st & \tr{\mat R + \math R - 2 (\math R^{1/2} \mat R \math R^{1/2})^{1/2} } \le \epsilon^2_0 \\
        & \mat R \succeq \mat 0,~\mat R_x \succ \mat 0.
    \end{array}
\end{equation}
The first constraint in the above display is a particularization of the Wasserstein distance between $\cal {CN}(\mat 0, \mat R)$ and $\cal {CN}(\mat 0, \math R)$. 

Problem \eqref{eq:BF-dist-robust-max-wasserstein} is a nonlinear positive semi-definite program (P-SDP). However, we can give it a linear reformulation.
\begin{proposition}\label{prop:wasserstein-R-x}
Problem \eqref{eq:BF-dist-robust-max-wasserstein} can be equivalently reformulated into a linear P-SDP
\begin{equation}\label{eq:BF-dist-robust-max-wasserstein-sol}
    \begin{array}{cl}
        \displaystyle \max_{\mat R, \mat V, \mat U} & \Tr [\mat R_s - \mat V] \\
        \st 
        &\left[
        \begin{array}{cc}
          \mat V   &  \mat R^\H_{xs}\\
          \mat R_{xs}  & \mat R_x
        \end{array}
        \right] \succeq \mat 0  \\
        &\tr{\mat R + \math R - 2 \mat U} \le \epsilon^2_0 \\
        &\left[
        \begin{array}{cc}
          \math R^{1/2} \mat R \math R^{1/2}   &  \mat U\\
           \mat U  & \mat I_{N+M}
        \end{array}
        \right] \succeq \mat 0  \\
        &\mat R \succeq \mat 0,~\mat R_x \succ \mat 0,~\mat V \succeq \mat 0,~\mat U \succeq \mat 0.
    \end{array}
\end{equation}
\end{proposition}
\begin{proof}
This is by applying the Schur complement. \stp
\end{proof}

Complex-valued linear P-SDP can be solved using, e.g., the YALMIP solver.\footnote{See \url{https://yalmip.github.io/inside/complexproblems/}.}

Suppose that $\mat R^\star$ solves \eqref{eq:BF-dist-robust-max-wasserstein-sol}. The corresponding Wasserstein distributionally robust beamformer is given as
\begin{equation}\label{eq:wasserstein-BF-R-x}
    \mat W^\star_{\text{DR-Wasserstein}} = \mat R^{\star \H}_{xs} \mat R^{\star -1}_x.
\end{equation}

Next, we separately investigate the uncertainties in $\math R_s$ and $\math R_v$. From \eqref{eq:BF-dist-robust-max-explicit}, we have
\begin{equation}\label{eq:BF-dist-robust-max-explicit-waserstein}
    \begin{array}{cl}
        \displaystyle \max_{\mat R_s, \mat R_v} &\Tr \big[\mat R_s -\mat R_s \mat H^\H (\mat H \mat R_s \mat H^\H + \mat R_v)^{-1} \mat H \mat R_s \big] \\
        \st 
        & \tr{\mat R_s + \math R_s - 2 (\math R_s^{1/2} \mat R_s \math R_s^{1/2})^{1/2} } \le \epsilon^2_1 \\
        & \tr{\mat R_v + \math R_v - 2 (\math R_v^{1/2} \mat R_v \math R_v^{1/2})^{1/2} } \le \epsilon^2_2 \\
        & \mat R_s \succeq \mat 0,~\mat R_v \succeq \mat 0,
    \end{array}
\end{equation}
where we ignore the uncertainty of $\mat H$ for technical tractability. Problem \eqref{eq:BF-dist-robust-max-explicit-waserstein} can be transformed into a linear P-SDP using a similar technique as in Proposition \ref{prop:wasserstein-R-x}. One can just introduce an inequality $\mat U \succeq \mat R_s \mat H^\H (\mat H \mat R_s \mat H^\H + \mat R)^{-1} \mat H \mat R_s$ and the objective function will become $\tr{\mat R_s - \mat U}$.

Suppose that $(\mat R_s^\star, \mat R_v^\star)$ solves \eqref{eq:BF-dist-robust-max-explicit-waserstein}. The corresponding Wasserstein distributionally robust beamformer is given as
\begin{equation}\label{eq:wasserstein-BF-M-R}
    \mat W^\star_{\text{DR-Wasserstein-Individual}} = \mat R_s^\star \mat H^\H [\mat H \mat R_s^\star \mat H^\H + \mat R_v^\star]^{-1}.
\end{equation}

\subsubsection{F-Norm Distributionally Robust Beamforming} Under the $F$-norm, we just need to replace the Wasserstein distance. 
To be specific, \eqref{eq:BF-dist-robust-max-wasserstein} becomes 
\begin{equation}\label{eq:BF-dist-robust-max-F-norm}
    \begin{array}{cl}
        \displaystyle \max_{\mat R} &\Tr \big[ -\mat R^{\H}_{xs}\mat R^{-1}_{x} \mat R_{xs} + \mat R_{s}\big] \\
        \st & \tr{(\mat R - \math R)^\H(\mat R - \math R)} \le \epsilon^2_0 \\
        & \mat R \succeq \mat 0,~\mat R_x \succ \mat 0.
    \end{array}
\end{equation}

The linear reformulation of the above display is given in the proposition below.

\begin{proposition}\label{prop:F-norm-R-x}
    The nonlinear P-SDP \eqref{eq:BF-dist-robust-max-F-norm} can be equivalently reformulated into a linear P-SDP
    \begin{equation}\label{eq:BF-dist-robust-max-F-norm-linear}
        \begin{array}{cl}
            \displaystyle \max_{\mat R, \mat V, \mat U} & \Tr [\mat R_s - \mat V] \\
            \st 
            &\left[
            \begin{array}{cc}
              \mat V   &  \mat R^\H_{xs}\\
              \mat R_{xs}  & \mat R_x
            \end{array}
            \right] \succeq \mat 0  \\
            & \tr{\mat U} \le \epsilon^2_0, \\
            
            & \left[
            \begin{array}{cc}
              \mat U   &  (\mat R - \math R)^\H\\
              (\mat R - \math R)   &  \mat I_{N+M}
            \end{array}
            \right] \succeq \mat 0, \\

            &\mat R \succeq \mat 0,~\mat R_x \succ \mat 0,~\mat V \succeq \mat 0,~\mat U \succeq \mat 0.
        \end{array}
    \end{equation}
\end{proposition}
\begin{proof}
This is by applying the Schur complement.
\stp
\end{proof}

\subsection{Multi-Frame Case: Dynamic Channel Evolution}\label{subsec:BF-multi-frame}
Each frame contains a pilot block used for beamformer design. Although the channel state information (CSI) may change from one frame to another, the CSI between the two consecutive frames is highly correlated. This correlation can benefit beamformer design across multiple frames. Suppose that $\{(\vec s_1, \vec x_1), (\vec s_2, \vec x_2), \ldots, (\vec s_L, \vec x_L)\}$ is the training data in the current frame and $\{(\vec s^\prime_1, \vec x^\prime_1), (\vec s^\prime_2, \vec x^\prime_2), \ldots, (\vec s^\prime_L, \vec x^\prime_L)\}$ is the history data in the immediately preceding frame. In such a case, the distributional difference between $\Ph_{\rvec x, \rvec s}$ and $\Ph_{\rvec x^{\prime}, \rvec s^{\prime}}$ is upper bounded, that is,
$
d(\Ph_{\rvec x, \rvec s},~\Ph_{\rvec x^{\prime}, \rvec s^{\prime}}) \le \epsilon^{\prime}
$ 
for some proper distance $d$ and a real number $\epsilon^{\prime} \ge 0$ where $\Ph_{\rvec x, \rvec s} \defeq \frac{1}{L} \sum^L_{i=1} \delta_{(\vec x_i, \vec s_i)}$ and $\Ph_{\rvec x^{\prime}, \rvec s^{\prime}} \defeq \frac{1}{L} \sum^L_{i=1} \delta_{(\vec x^{\prime}_i, \vec s^{\prime}_i)}$. 

Since a beamformer $\mat W = \cal F(\P_{\rvec x, \rvec s})$ is a continuous functional $\cal F (\cdot)$ of data distribution $\P_{\rvec x, \rvec s}$, cf. \eqref{eq:opt-BF}, we have 
$
\|\mat W - \mat W^\prime\|_F = \|\cal F(\Ph_{\rvec x, \rvec s}) - \cal F(\Ph_{\rvec x^{\prime}, \rvec s^{\prime}})\|_F 
\le C \cdot d(\Ph_{\rvec x, \rvec s},~\Ph_{\rvec x^{\prime}, \rvec s^{\prime}}) 
\le \epsilon
$
for some positive constant $C \ge 0$ and upper bound $\epsilon \ge 0$ where $\mat W^\prime$ is the beamformer associated with $\Ph_{\rvec x^{\prime}, \rvec s^{\prime}}$ in the previous frame. Therefore, the beamforming problem \eqref{eq:opt-BF-explicit} becomes a constrained problem
\[
\begin{array}{cl}
\displaystyle \min_{\mat W} & \Tr \big[\mat W \mat R_{x} \mat W^\H - \mat W \mat R_{xs} - \mat R^\H_{xs} \mat W^\H + \mat R_{s}\big] \\
\st & \Tr[\mat W - \mat W^\prime][\mat W - \mat W^\prime]^{\H} \le \epsilon^2.
\end{array}
\]
By the Lagrange duality theory, it is equivalent to
\begin{equation}\label{eq:opt-BF-explicit-multi-frame}
\begin{array}{l}
\displaystyle \min_{\mat W} \Tr \big[\mat W \mat R_{x} \mat W^\H - \mat W \mat R_{xs} - \mat R^\H_{xs} \mat W^\H + \mat R_{s}\big] + \\
\quad \quad \quad \quad \quad \quad \lambda \cdot \Tr[\mat W - \mat W^\prime][\mat W - \mat W^\prime]^{\H} \\

= \displaystyle \min_{\mat W} \Tr \big[\mat W (\mat R_{x} + \lambda \mat I_N) \mat W^\H - \mat W (\mat R_{xs} + \lambda \mat W^{\prime \H}) - \\
\quad \quad \quad \quad \quad \quad (\mat R_{xs} + \lambda \mat W^{\prime \H})^\H \mat W^\H + (\mat R_{s} + \lambda \mat W^\prime \mat W^{\prime \H}) \big],
\end{array}
\end{equation}
for some $\lambda \ge 0$. As a result, we have the Wiener beamformer for the multi-frame case, where we can treat $\mat W^\prime$ as a \bfit{prior knowledge} of $\mat W$.

\begin{claim}[Multi-Frame Beamforming]\label{claim:BF-multi-frame}
The Wiener beamformer for the multi-frame case is given by
 \begin{equation}\label{eq:BF-multi-frame}
\begin{array}{cl}
    \mat W^\star_{\text{Wiener-MF}} &= [\mat R_{xs} + \lambda \mat W^{\prime \H}]^\H[\mat R_{x} + \lambda \mat I_N]^{-1} \\
    &= [\mat R_s \mat H^\H + \lambda \mat W^{\prime}] [\mat H \mat R_s \mat H^\H + \mat R_v + \lambda \mat I_N]^{-1},
\end{array}
\end{equation}
where $\lambda \ge 0$ is a tuning parameter to control the similarity between $\mat W$ and $\mat W^\prime$. Specifically, if $\lambda$ is large, $\mat W$ must be close to $\mat W^\prime$; if $\lambda$ is small, $\mat W$ can be far away from $\mat W^\prime$. \stp
\end{claim}

With the result in Claim \ref{claim:BF-multi-frame}, \eqref{eq:BF-dist-robust-max} becomes 
\begin{equation}\label{eq:BF-dist-robust-max-multi-frame}
    \begin{array}{cl}
        \displaystyle \max_{\mat R} &\Tr \big[ -(\mat R_{xs} + \lambda \mat W^{\prime \H})^\H \cdot (\mat R_{x} + \lambda \mat I_N)^{-1} \cdot \\
        & \quad \quad \quad (\mat R_{xs} + \lambda \mat W^{\prime \H}) + (\mat R_{s} + \lambda \mat W^\prime \mat W^{\prime \H})\big] \\
        \st & d_0(\mat R, {\math R}) \le \epsilon_0, \\
        & 
        \mat R \succeq \mat 0,
    \end{array}
\end{equation}
whose objective function is monotonically increasing in $\mat R$. 

The remaining distributional robustness modeling and analyses against the uncertainties in $\mat R$ are technically straightforward, and therefore, we omit them here. Upon using the diagonal-loading method on $\mat R$, a distributionally robust beamformer for the multi-frame case is
\[
\begin{array}{l}
    \mat W^\star_{\text{DR-Wiener-MF}} = [\math R_{xs} + \lambda \mat W^{\prime \H}]^\H \cdot [\math R_x + \lambda \mat I_N + \epsilon_0 \mat I_N]^{-1},
\end{array}
\]
where $\epsilon_0$ is an uncertainty quantification parameter for $\mat R$.

\section{Distributionally Robust Nonlinear Estimation}\label{sec:nonlin-esti}

For the convenience of the technical treatment, we study the estimation problem in real spaces. Nonlinear estimators, which are suitable for non-Gaussian $\P_{\rvec x, \rvec s}$, are to be limited in reproducing kernel Hilbert spaces and feedforward multi-layer neural network function spaces. 

\subsection{Reproducing Kernel Hilbert Spaces}\label{subsec:RKHS}

\subsubsection{General Framework and Concrete Examples}

As a standard treatment in machine learning, we use the partial pilot data $\{\vecul x_1, \vecul x_2, \ldots, \vecul x_L\}$ to construct the reproducing kernel Hilbert spaces, and use the whole pilot data $\{(\vecul x_1, \vecul s_1), (\vecul x_2, \vecul s_2), \ldots, (\vecul x_L, \vecul s_L)\}$ to train the optimal estimator in an RKHS.

With the $\mat W$-linear representation of $\vec \phi(\cdot)$ in \eqref{eq:non-lin-func-multi-RKHS}, i.e., $\vec \phi(\cdot) = \mat W \vec \varphi(\cdot)$, the distributionally robust estimation problem \eqref{eq:opt-esti-dist-robust} becomes 
\begin{equation}\label{eq:opt-esti-dist-robust-RKHS}
    \min_{\mat W \in \R^{2M \times L}} \max_{\P_{\rvecul x, \rvecul s} \in \cal U_{\rvecul x, \rvecul s}} \Tr \E_{\rvecul x, \rvecul s}[\mat W \cdot \vec \varphi(\rvecul x) - \rvecul s][\mat W \cdot \vec \varphi(\rvecul x) - \rvecul s]^\T.
\end{equation}
The proposition below reformulates and solves \eqref{eq:opt-esti-dist-robust-RKHS}.

\begin{proposition}\label{prop:RHKS-reformulation}
Let $\mat K$ denote the \textit{kernel matrix} associated with the kernel function $\ker(\cdot, \cdot)$ whose $(i,j)$-entry is defined as
\[
    \mat K_{i,j} \defeq \ker(\vecul x_i, \vecul x_j),~~~\forall i,j \in [L].
\]
Let $\rvecul z \defeq \vec \varphi(\rvecul x)$. Then, the distributionally robust $\rvecul x$-nonlinear estimation problem \eqref{eq:opt-esti-dist-robust-RKHS} can be rewritten as a distributionally robust $\rvecul z$-linear beamforming problem as 
\begin{equation}\label{eq:BF-dist-robust-z}
    \begin{array}{cl}
        \displaystyle \min_{\mat W} \max_{\mat R_{\ul z}, \mat R_{\ul{zs}}, \mat R_{\ul s}} &\Tr \big[\mat W \mat R_{\ul z} \mat W^\T - \mat W \mat R_{\ul{zs}} - \mat R^\T_{\ul{zs}} \mat W^\T + \mat R_{\ul s}\big] \\
        \st 
        & d_0\left(
        \left[
        \begin{array}{cc}
           \mat R_{\ul z}  &  \mat R_{\ul{zs}} \\
           \mat R^\T_{\ul{zs}}  &  \mat R_{\ul s}
        \end{array}
        \right],
        \left[
        \begin{array}{cc}
           \math R_{\ul z}  &  \math R_{\ul{zs}} \\
           \math R^\T_{\ul{zs}}  &  \math R_{\ul s}
        \end{array}
        \right]
        \right) \le \epsilon_0,\\
        & 
        \left[
        \begin{array}{cc}
           \mat R_{\ul z}  &  \mat R_{\ul{zs}} \\
           \mat R^\T_{\ul{zs}}  &  \mat R_{\ul s}
        \end{array}
        \right] \succeq \mat 0,
    \end{array}
\end{equation}
where 
$\math R_{\ul z} = \frac{1}{L} \mat K^2$, $\math {R}_{\ul{zs}} = \frac{1}{L} \mat K \matul S^\T$, $\math {R}_{\ul s} = \frac{1}{L} \matul S \matul S^\T$, and $\matul S \defeq [\Re \mat S;~\Im \mat S] = [\vecul s_1, \vecul s_2, \ldots, \vecul s_L]$.
In addition, the strong min-max property holds for \eqref{eq:BF-dist-robust-z}: i.e., the order of $\min$ and $\max$ can be exchanged provided that the first constraint is compact convex. As a result, given every pair of $(\mat R_{\ul z}, \mat R_{\ul{zs}}, \mat R_{\ul s})$, the optimal Wiener beamformer is 
\begin{equation}\label{eq:opt-estimator-RKHS-real}
    \mat W^\star_{\text{RKHS}} = \mat R^\T_{\ul{zs}} \cdot \mat R^{-1}_{\ul z} 
\end{equation}
which transforms \eqref{eq:BF-dist-robust-z} to
\begin{equation}\label{eq:BF-dist-robust-z-max}
    \begin{array}{cl}
        \displaystyle \max_{\mat R_{\ul z}, \mat R_{\ul{zs}}, \mat R_{\ul s}} &\Tr \big[ - \mat R^\T_{\ul{zs}} \mat R^{-1}_{\ul z} \mat R_{\ul{zs}} + \mat R_{\ul s}\big] \\
        \st 
        & d_0\left(
        \left[
        \begin{array}{cc}
           \mat R_{\ul z}  &  \mat R_{\ul{zs}} \\
           \mat R^\T_{\ul{zs}}  &  \mat R_{\ul s}
        \end{array}
        \right],~
        \left[
        \begin{array}{cc}
           \math R_{\ul z}  &  \math R_{\ul{zs}} \\
           \math R^\T_{\ul{zs}}  &  \math R_{\ul s}
        \end{array}
        \right]
        \right) \le \epsilon_0,\\
        & 
        \left[
        \begin{array}{cc}
           \mat R_{\ul z}  &  \mat R_{\ul{zs}} \\
           \mat R^\T_{\ul{zs}}  &  \mat R_{\ul s}
        \end{array}
        \right] \succeq \mat 0,~~~\mat R_{\ul z} \succ \mat 0.
    \end{array}
\end{equation}
\end{proposition}
\begin{proof}
    Treating $[\rvecul z; \rvecul s]$ as, or approximating $[\rvecul z; \rvecul s]$ using, a joint Gaussian random vector due to the linear estimation relation $\hat{\rvecul s} = \mat W \rvecul z$ in RKHS [cf. \eqref{eq:opt-esti-dist-robust-RKHS}], then the results in Lemma \ref{lem:lin-BF-dist-robust-dual} apply. For details, see Appendix \ref{append:BF-dist-robust-z}. \stp
\end{proof}

In \eqref{eq:BF-dist-robust-z}, $d_0$ defines a matrix similarity measure to quantify the uncertainty of the covariance matrix of $[\rvecul z; \rvecul s]$, and $\epsilon_0 \ge 0$ quantifies the uncertainty level. Proposition \ref{prop:RHKS-reformulation} reveals the benefit of the kernel trick \eqref{eq:non-lin-func-multi-RKHS}, that is, the possibility to represent a nonlinear estimation problem as a linear one.

The claim below summarizes the solution of \eqref{eq:opt-esti-dist-robust} in the RKHS induced by the kernel function $\ker(\cdot, \cdot)$.

\begin{claim}\label{claim:opt-estimator-RKHS}
    Suppose that $(\mat R^\star_{\ul z}, \mat R^\star_{\ul{zs}}, \mat R^\star_{\ul s})$ solves \eqref{eq:BF-dist-robust-z-max}. Then the optimal estimator of \eqref{eq:opt-esti-dist-robust} in the RKHS induced by $\ker(\cdot, \cdot)$ is given by
    \begin{equation}\label{eq:opt-estimator-RKHS}
        \vec \phi^\star(\rvec x) = \mat \Gamma_M \cdot \mat R^{\star\T}_{\ul{zs}} \cdot \mat R^{\star-1}_{\ul z} \cdot \vec \varphi(\rvecul x),
    \end{equation}
    where $\rvecul x = [\Re \rvec x;~\Im \rvec x]$ is the real-space representation of $\rvec x$, $\mat \Gamma_M \defeq [\mat I_M, \mat J_M]$ is defined in Subsection \ref{subsec:notation}, and
    \[
        \vec \varphi(\rvecul x) \defeq
        \left[
        \begin{array}{c}
            \ker(\rvecul x, \vecul x_1) \\
            \ker(\rvecul x, \vecul x_2) \\
            \vdots \\
            \ker(\rvecul x, \vecul x_L)
        \end{array}
        \right].
    \]
    In addition, the corresponding worst-case estimation error covariance is
    \begin{equation}\label{eq:opt-estimator-RKHS-cov}
        \mat \Gamma_M \cdot \big[ - \mat R^{\star\T}_{\ul{zs}} \mat R^{\star-1}_{\ul z} \mat R^\star_{\ul{zs}} + \mat R^\star_{\ul s}\big] \cdot \mat \Gamma_M^\H,
    \end{equation}
which upper bounds the true estimation error covariance. \stp
\end{claim}

Concrete examples of Claim \ref{claim:opt-estimator-RKHS} are given as follows.
\begin{example}[Kernelized Diagonal Loading]\label{exam:opt-estimator-RKHS}
By using the trimmed diagonal-loading uncertainty set for $\mat R_{\ul z}$, i.e.,
\[
 \math R_{\ul z} - \epsilon_0 \mat I_L \preceq \mat R_{\ul z} \preceq \math R_{\ul z} + \epsilon_0 \mat I_L,
\]
we have the kernelized diagonal loading method
\begin{equation}\label{eq:kernel-ridge-sol}
    \vec \phi^\star(\rvec x) = \mat \Gamma_M \cdot \frac{1}{L} \matul S \mat K \cdot \left(\frac{1}{L}\mat K^2 + \epsilon_0 \mat I_L\right)^{-1} \cdot \vec \varphi(\rvecul x),
\end{equation}
which is obtained at the upper bound of $\mat R_{\ul z}$. Furthermore, in this case, the distributionally robust formulation \eqref{eq:opt-esti-dist-robust-RKHS} is equivalent to a squared-$F$-norm-regularized formulation
\begin{equation}\label{eq:kernel-ridge}
\begin{array}{l}
    \displaystyle \min_{\mat W} \Tr \E_{(\rvecul x, \rvecul s) \sim \Ph_{\rvecul x, \rvecul s}}[\mat W \cdot \vec \varphi(\rvecul x) - \rvecul s][\mat W \cdot \vec \varphi(\rvecul x) - \rvecul s]^\T + \\
    \quad \quad \quad \quad \quad \quad \quad \quad \quad \quad \epsilon_0 \cdot \Tr [\mat W \mat W^\T],
\end{array}
\end{equation}
which can be proven by replacing $\mat R_{\ul z}$  in \eqref{eq:BF-dist-robust-z} with its upper bound. \stp
\end{example}

\begin{example}[Kernelized Eigenvalue Thresholding]\label{exam:opt-estimator-RKHS-eigen-thres}
The kernelized eigenvalue thresholding method can be designed in analogy to Example \ref{exam:eigen-thres}. The two key steps are to obtain the eigenvalue decomposition of $\math R_{\ul z} = \mat K^2/L$ and then lift the eigenvalues; cf. \eqref{eq:eigen-thres}.
\stp
\end{example}

In addition, Example \ref{exam:opt-estimator-RKHS} motivates the following important theorem for statistical machine learning.

\begin{theorem}[Kernel Ridge Regression and Kernel  Tikhonov Regularization]\label{thm:kernel-ridge}
Consider the nonlinear regression problem
\[
    \rvec s = \vec \phi(\rvec x) + \rvec e,
\]
and the distributionally robust estimator of $\vec \phi(\rvecul x) = \mat W \cdot \vec \varphi(\rvecul x)$ in the RKHS induced by the kernel function $\ker(\cdot, \cdot)$, i.e.,
\[
    \min_{\mat W \in \R^{2M \times L}} \max_{\P_{\rvecul x, \rvecul s} \in \cal U_{\rvecul x, \rvecul s}} \Tr \E_{\rvecul x, \rvecul s}[\mat W \cdot \vec \varphi(\rvecul x) - \rvecul s][\mat W \cdot \vec \varphi(\rvecul x) - \rvecul s]^\T.
\]
Supposing that only the second-order moment of $\rvecul z \defeq \vec \varphi(\rvecul x)$ is uncertain and quantified as
\[
{\math R}_{\ul z} - \epsilon_0 \mat I_L \preceq {\mat R}_{\ul z} \preceq {\math R}_{\ul z} + \epsilon_0 \mat I_L,
\] 
then the distributionally robust estimator of $\mat W$ becomes a kernel ridge regression method \eqref{eq:kernel-ridge}. The regularization term in \eqref{eq:kernel-ridge} becomes the Tikhonov regularizer $\Tr [\mat W \mat F \mat W^\T]$ if 
\[
{\math R}_{\ul z} - \epsilon_0 \mat F \preceq {\mat R}_{\ul z} \preceq {\math R}_{\ul z} + \epsilon_0 \mat F
\]
for some $\mat F \succeq \mat 0$.
\end{theorem}
\begin{proof}
See Example \ref{exam:opt-estimator-RKHS}; cf. Theorem \ref{thm:ridge-regression}.
\stp
\end{proof}

Theorem \ref{thm:kernel-ridge} gives the kernel ridge regression an interpretation of distributional robustness. The usual choice of $\mat F$ in Theorem \ref{thm:kernel-ridge} is the $L$-divided kernel matrix $\mat K/L$; see, e.g., \cite[Eq. (4)]{vu2015understanding}, \cite[Eqs. (15.110) and (15.113)]{murphy2012machine}. As a result, from \eqref{eq:kernel-ridge-sol}, we have
\begin{equation}\label{eq:kernel-ridge-sol-2}
\vec \phi^\star(\rvec x) = \mat \Gamma_M \cdot \matul S \cdot \left(\mat K + \epsilon_0 \mat I_L\right)^{-1} \cdot \vec \varphi(\rvecul x),
\end{equation}
which is another type of kernel ridge regression (i.e., a new kernelized diagonal-loading method). 

In analogy to Corollary \ref{cor:any-norm-regularization}, the following corollary motivated from \eqref{eq:kernel-ridge} is immediate.
\begin{corollary}
The following squared-norm-regularized  method in RKHSs can combat the distributional uncertainty:
\begin{equation}\label{eq:kernel-any-norm}
\begin{array}{l}
    \displaystyle \min_{\mat W} \Tr \E_{(\rvecul x, \rvecul s) \sim \Ph_{\rvecul x, \rvecul s}}[\mat W \cdot \vec \varphi(\rvecul x) - \rvecul s][\mat W \cdot \vec \varphi(\rvecul x) - \rvecul s]^\T + \\
    \quad \quad \quad \quad \quad \quad \quad \quad \quad \quad \lambda \cdot \|\mat W\|^2,
\end{array}
\end{equation}
for any matrix norm $\|\cdot\|$; cf. Corollary \ref{cor:any-norm-regularization}. \stp
\end{corollary}

Moreover, in analogy to Corollary \ref{cor:adv-learning-linear}, the following corollary is immediate.

\begin{corollary}[Data Augmentation for Kernel Regression]\label{cor:adv-learning-linear-RKHS}
Consider the nonlinear regression problem in Theorem \ref{thm:kernel-ridge}. Its data-perturbed counterpart can be constructed by taking into account the data perturbation vectors $(\rvec \Delta_{\ul s}, \rvec \Delta_{\ul z})$. Suppose that $\rvec \Delta_{\ul z}$ is uncorrelated with $\rvecul z$, with $\rvecul s$, and with $\rvec \Delta_{\ul s}$; in addition, $\rvec \Delta_{\ul s}$ is uncorrelated with $\rvecul z$. If the second-order moment of $\rvec \Delta_{\ul z}$ is upper bounded by $\epsilon_0 \mat I_L$, then the distributionally robust estimator of $\mat W$ becomes a kernel ridge regression (i.e., squared-$F$-norm regularized) method \eqref{eq:kernel-ridge}. The regularization term becomes $\Tr \big[\mat W \mat F \mat W^\H\big]$, which is known as the Tikhonov regularizer, if the second-order moment of $\rvec \Delta_{\ul z}$ is upper bounded by $\epsilon_0 \mat F$ for some $\mat F \succeq \mat 0$.  \stp
\end{corollary}

General uncertainty sets using the Wasserstein distance or the $F$-norm, beyond the diagonal $\epsilon_0$-perturbation (cf. Example \ref{exam:opt-estimator-RKHS}), can be straightforwardly employed and the distributional robustness modeling and analyses remain routine; cf. Subsection \ref{subsec:BF-nontrivial-sets}. Hence, we omit them here. However, such complicated approaches are computationally prohibitive in practice when $L$ or $M$ is large.

\subsubsection{Multi-Frame Case: Dynamic Channel Evolution}\label{subsubsec:multi-frame-RKHS}

As in \eqref{eq:opt-BF-explicit-multi-frame}, the multi-frame formulation in RKHSs is
\begin{equation}\label{eq:kernel-ridge-multi-frame}
\begin{array}{l}
    \displaystyle \min_{\mat W \in \R^{2M \times L}} \Tr \E_{\rvecul x, \rvecul s}[\mat W \cdot \vec \varphi(\rvecul x) - \rvecul s][\mat W \cdot \vec \varphi(\rvecul x) - \rvecul s]^\T + \\
    \quad \quad \quad \quad \quad \quad \quad \quad \lambda \cdot \Tr [\mat W - \mat W'] [\mat W - \mat W']^\T,
\end{array}
\end{equation}
where $\mat W'$ denotes the beamformer in the immediately preceding frame and serves as a \bfit{prior knowledge} of $\mat W$.

\begin{claim}[Multi-Frame Estimation in RHKS]\label{claim:BF-multi-frame-RKHS}
The solution to \eqref{eq:kernel-ridge-multi-frame} is given by (cf. \eqref{eq:opt-estimator-RKHS-real})
 \begin{equation}\label{eq:BF-multi-frame-RKHS}
\begin{array}{cl}
    \mat W^\star_{\text{RKHS-MF}} &= [\mat R_{\ul{zs}} + \lambda \mat W^{\prime \T}]^\T [\mat R_{\ul z} + \lambda \mat I_L]^{-1} \\
    &= \left(\frac{1}{L} \matul S \mat K + \lambda \mat W^{\prime} \right) \cdot \left(\frac{1}{L}\mat K^2 + \lambda \mat I_L\right)^{-1},
\end{array}
\end{equation}
where $\lambda \ge 0$ is a tuning parameter to control the similarity between $\mat W$ and $\mat W^\prime$; cf. Claim \ref{claim:BF-multi-frame}. \stp
\end{claim}

The remaining distributional robustness modeling and analyses on \eqref{eq:kernel-ridge-multi-frame} against the uncertainties in $\math R_{\ul z}$, $\math R_{\ul{xz}}$, and $\math R_{\ul s}$ are technically straightforward; cf. Subsection \ref{subsec:BF-multi-frame}. Therefore, we omit them here.

\subsection{Neural Networks}\label{subsec:NN}
With the $\mat W$-parameterization $\vec \phi_{\mat W_{[R]}}(\rvecul x)$ of $\vec \phi(\rvecul x)$ in feedforward multi-layer neural networks, i.e., \eqref{eq:non-lin-func-NNFS}, the distributionally robust estimation problem \eqref{eq:opt-esti-dist-robust} becomes 
\begin{equation}\label{eq:opt-esti-dist-robust-NN}
    \min_{\mat W_{[R]}}~~\max_{\P_{\rvecul x, \rvecul s} \in \cal U_{\rvecul x, \rvecul s}} \Tr \E_{\rvecul x, \rvecul s}[\vec \phi_{\mat W_{[R]}}(\rvecul x) - \rvecul s][\vec \phi_{\mat W_{[R]}}(\rvecul x) - \rvecul s]^\T,
\end{equation}
where $\mat W_{[R]} \defeq \{\mat W_1,\mat W_2,\ldots,\mat W_R\}$ and $\vec \phi_{\mat W_{[R]}}(\rvecul x)$ is defined in \eqref{eq:non-lin-func-NNFS}. Problem \eqref{eq:opt-esti-dist-robust-NN} is highly nonlinear in both argument $\rvecul x$ and parameter $\mat W_{[R]}$, which is different from the case in reproducing kernel Hilbert spaces where the $\mat W$-linearization features. Hence, problem \eqref{eq:opt-esti-dist-robust-NN} is too complicated to solve to global optimality. According to \cite[Cor.~33]{shafieezadeh2019regularization}, under several technical conditions (plus the boundedness of the feasible region of $\mat W_{[R]}$), \eqref{eq:opt-esti-dist-robust-NN} is upper bounded by a spectral-norm-regularized empirical risk minimization problem
\begin{equation}\label{eq:opt-esti-regularized-NN}
\begin{array}{l}
    \displaystyle \min_{\mat W_{[R]}} \displaystyle \frac{1}{L} \sum^L_{i=1} \Tr[\vec \phi_{\mat W_{[R]}}(\vecul x_i) - \vecul s_i][\vec \phi_{\mat W_{[R]}}(\vecul x_i) - \vecul s_i]^\T + \\
    \quad \quad \quad \quad \quad \quad \displaystyle \lambda' \cdot \sum^R_{r = 1} \| \mat W_{r} \|_2,
\end{array}
\end{equation}
for some regularization coefficient $\lambda' \ge 0$, where $\|\cdot\|_2$ denotes the spectral norm of a matrix (i.e., the induced $2$-norm). Eq. \eqref{eq:opt-esti-regularized-NN} rigorously justifies the popular norm regularization method in training neural networks: By diminishing the upper bound \eqref{eq:opt-esti-regularized-NN} of \eqref{eq:opt-esti-dist-robust-NN}, the true error in \eqref{eq:opt-esti-dist-robust-NN} can be controlled from above. The regularized ERM problem \eqref{eq:opt-esti-regularized-NN} is reminiscent of the ridge regression and the kernel ridge regression methods in Theorems \ref{thm:ridge-regression} and \ref{thm:kernel-ridge} for distributional robustness in linear regression and RKHS linear regression, respectively. Supposing that $\mat W^\star_{[R]}$ is an (approximated, or sub-optimal) solution\footnote{Neural networks are hard to be globally optimized.} of \eqref{eq:opt-esti-regularized-NN}, then the distributionally robust optimal estimator of the transmitted signal $\rvec s$ can be obtained as
\[
\rvech s = \mat \Gamma_M \cdot \vec \phi_{\mat W^\star_{[R]}}(\rvecul x).
\]
Therefore, in training a neural network for wireless signal estimation, it is recommended to apply the norm regularization methods. Since norms on real spaces are equivalent, \eqref{eq:opt-esti-regularized-NN} can be further upper bounded by
\begin{equation}\label{eq:opt-esti-p-regularized-NN}
\begin{array}{l}
    \displaystyle \min_{\mat W_{[R]}} \displaystyle \frac{1}{L} \sum^L_{i=1} \Tr[\vec \phi_{\mat W_{[R]}}(\vecul x_i) - \vecul s_i][\vec \phi_{\mat W_{[R]}}(\vecul x_i) - \vecul s_i]^\T + \\
    \quad \quad \quad \quad \quad \quad \displaystyle \lambda \cdot \sum^R_{r = 1} \| \mat W_{r} \|,
\end{array}
\end{equation}
for any matrix norm $\|\cdot\|$ and some $\lambda \ge 0$; $\lambda$ depends on $\lambda'$ and $\|\cdot\|$. As a result, to achieve distributional robustness in training a neural network, any-norm-regularized learning method in \eqref{eq:opt-esti-p-regularized-NN} can be considered.


\section{Experiments}\label{sec:experiment}
We consider a point-to-point multiple-input-multiple-output (MIMO) wireless communication problem where the transmitter is located at $[0, 0]$ and the receiver is at $[500\text{m}, 450\text{m}]$. We randomly sample $25$ points according to the uniform distribution on the square of $[0, 500\text{m}] \times [0, 500\text{m}]$ to denote the scatters' positions; i.e., there exist $26$ radio paths. All the source data and codes are available online at GitHub with thorough implementation comments: \url{https://github.com/Spratm-Asleaf/DRRC}. In this section, we only present major experimental setups and results; readers can use the shared source codes to explore (or verify) minor ones.

The following eleven methods are implemented in the experiments: 1) \textbf{Wiener}: Wiener beamformer \eqref{eq:BF-Wiener}, upper expression; 2) \textbf{Wiener-DL}: Wiener beamformer with diagonal loading \eqref{eq:DRBF-DL}, upper expression; 3) \textbf{Wiener-DR}: Distributionally robust Wiener beamformer \eqref{eq:wasserstein-BF-R-x} and \eqref{eq:BF-dist-robust-max-F-norm-linear}; 4) \textbf{Wiener-CE}: Channel-estimation-based Wiener beamformer \eqref{eq:BF-Wiener}, lower expression; 5) \textbf{Wiener-CE-DL}: Channel-estimation-based Wiener beamformer with diagonal loading \eqref{eq:DRBF-DL}, lower expression; 6) \textbf{Wiener-CE-DR}: Distributionally robust channel-estimation-based Wiener beamformer \eqref{eq:diag-loading-using-H} and \eqref{eq:DRBF-GDL}; 7) \textbf{Capon}: Capon beamformer \eqref{eq:diag-loading-capon} for $\epsilon_0 = 0$; 8) \textbf{Capon-DL}: Capon beamformer with diagonal loading \eqref{eq:diag-loading-capon}; 9) \textbf{ZF}: Zero-forcing beamformer where $\mat W_{\text{ZF}} \defeq (\math H^\H \math H)^{-1}\math H^\H$ and $\math H$ denotes the estimated channel matrix; 10) \textbf{Kernel}: Kernel receiver \eqref{eq:opt-estimator-RKHS} with $\epsilon_0 = 0$ in \eqref{eq:BF-dist-robust-z-max}; and 11) \textbf{Kernel-DL}: Kernel receiver with diagonal loading \eqref{eq:kernel-ridge-sol-2}. Note that the diagonal-loading-based methods are particular cases of distributionally robust combiners; see, e.g., Corollary \ref{cor:solutions-under-moments} and Example \ref{exam:opt-estimator-RKHS}. The deep-learning-based (DL-based) methods in Subsection \ref{subsec:NN} are not implemented in this section because they have been deeply studied in our previous publications, e.g., \cite{ye2017power,van2022transfer}; we only comment on the advantages and disadvantages of DL-based methods compared with the listed eleven methods in Section \ref{sec:conclusion} (Conclusions).

When covariance matrix $\mat R_s$ of transmitted signal $\rvec s$ is unknown for the receiver (e.g., in ISAC systems, $\mat R_s$ needs to vary from one frame to another for sensing), $\mat R_s$ is estimated by the sample covariance matrix $\math R_s = \mat S \mat S^\H/L$. The channel matrix $\mat H$ is estimated using the minimum mean-squared error method, i.e., $\math H = \mat X \mat S^\H (\mat S \mat S^\H)^{-1}$. Covariance matrix $\mat R_v$ of channel noise $\rvec v$ is estimated using the least-square method, i.e., $\math R_v = (\mat X - \math H \mat S)(\mat X - \math H \mat S)^\H/L$. The matrices $\math R_s$, $\math H$, and $\math R_v$ are therefore uncertain compared to their true (but unknown; possibly time-varying) values $\mat R_s$, $\mat H$, and $\mat R_v$, respectively. The matrices $\math R_s$, $\math H$, and $\math R_v$ are used in beamformers such as the channel-estimation-based Wiener beamformer \eqref{eq:DRBF-DL}, the Capon beamformer, and the zero-forcing beamformer.

The combiners are determined on the training data set (i.e., pilot data). The performance evaluation method of combiners is mean-squared estimation error (MSE) on the test data set (i.e., non-pilot communication data): to be specific, $\|\mat S_{\text{test}} - \math S_{\text{test}}\|^2_F/(M \times L_{\text{test}})$ where $\mat S_{\text{test}} \in \C^{M \times L_{\text{test}}}$ is the test data block, $\math S_{\text{test}}$ is its estimate, and $L_{\text{test}}$ is the length of non-pilot test data units. 
As data-driven machine learning methods, all parameters (e.g., uncertainty quantification coefficients $\epsilon$'s) of combiners can be tuned using the popular cross-validation (e.g., one-shot cross-validation) method. The parameters can also be empirically tuned to save training times because cross-validation imposes a significant computational burden. 
This article mainly uses the empirical tuning method (i.e., trial-and-error) to tune each combiner to achieve its best average performance. For each test case, the MSE performances are averaged on $250$ Monte--Carlo episodes.

We consider an experimental scenario where impulse channel noises exist; i.e., the channel is non-Gaussian so linear beamformers are no longer sufficient. (Complementary experimental setups and results can be seen in online supplementary
materials.) The detailed setups are as follows. The transmitter has four antennas (i.e., $M = 4$) with unit transmit power; without loss of generality, each antenna is assumed to emit continuous-valued complex Gaussian signals. The receiver has eight antennas (i.e., $N = 8$). 
The SNR is $-10$dB, which is a challenging situation. 
The channel has impulse noises: i.e., in $L$ received signals (i.e., $[\rvec x_1, \rvec x_2, \ldots, \rvec x_L]$) that are contaminated by usual complex Gaussian channel noises, $10\%$ of them are also contaminated by uniform noises with the maximum amplitude of $1.5$, which is a relatively large value compared to the amplitude of the usual Gaussian channel noises. We assume that a communication frame contains $500$ non-pilot data units; i.e., $L_{\text{test}} = 500$. The experimental results are shown in Tables \ref{tab:scenario-2-1}$\sim$\ref{tab:scenario-2-6}, from which the following main points can be outlined.
\begin{enumerate}
    \item A larger number of pilot data benefits the estimation performances of wireless signals. 
    
    \item The diagonal loading operation can significantly improve the estimation performances especially when the pilot data size is relatively small.

    \item Since the signal model under impulse channel noises is no longer linear Gaussian, the optimal combiner in the MSE sense must be nonlinear. Therefore, the Kernel and the Kernel-DL methods have the potential to outperform other linear beamformers, i.e., to \textit{suppress outliers}. However, in practice, the non-robust Kernel method may undergo numerical instability in calculating the inverse of the kernel matrix $\mat K$. Therefore, its actual MSEs are not necessarily smaller than those of linear beamformers. Nevertheless, the robust Kernel-DL method consistently outperforms all other beamformers.

    \item Distributionally robust combiners (including diagonal-loading ones) can combat the adverse effect introduced by the limited pilot size and several types of uncertainties in the signal model (e.g., outliers). To be specific, for example, all diagonal-loading combiners can outperform their original non-diagonal-loading counterparts; cf. the Wiener and the Wiener-DL methods, the Wiener-CE and the Wiener-CE-DL methods, the Capon and the Capon-DL methods, and the Kernel and the Kernel-DL methods. In addition, the Wiener-DR beamformer \eqref{eq:BF-dist-robust-max-F-norm-linear} using the $F$-norm uncertainty set has the potential to outperform the Wiener-DL beamformer \eqref{eq:DRBF-DL} that employs the simple uncertainty set \eqref{eq:uncertainty-set-diag-loading}.
    
    \item Although the Wiener-DR beamformer has the potential to work better than the Wiener-DL beamformer, it has a significant computational burden, which may not be suitable for timely use in practice especially when the computing resources are limited. Hence, the Wiener-DL beamformer is practically promising because it can provide an excellent balance between the computational burden and the actual performance. 
\end{enumerate}

\textit{Remarks on Parameter Tuning}: From experiments, we find that the uncertainty quantification coefficients $\epsilon$'s (e.g., in diagonal loading) can be neither too large nor too small. When $\epsilon$'s are too large, the combiners become overly conservative, while when $\epsilon$'s are too small, the combiners cannot offer sufficient robustness against data scarcity and model uncertainties. In both cases of inappropriate $\epsilon$'s, the performances of combiners degrade significantly. Therefore, $\epsilon$'s must be carefully tuned in practice, and a rigorous method to tune $\epsilon$'s can be the cross-validation method on the training data set (i.e., the pilot data set). If practitioners just pursue satisfaction rather than optimality, the empirical tuning method is recommended to save training times.

\begin{table}[!htbp]
\centering
\caption{Experimental Results (Pilot Size = 10)}
\label{tab:scenario-2-1}
\begin{tabular}{lrr|lrr}
\bottomrule[0.8pt]
Combiner  & MSE & Time & Combiner  & MSE & Time \\
\specialrule{0.7pt}{0pt}{0pt}
Wnr           & 3.30      & 1.49e-04      & Wnr-DL          & 2.11    & 9.81e-06 \\
\hline
Wnr-DR        & 1.97      & \red{3.16e+00}      & Wnr-CE          & 3.30    & 4.59e-05 \\
\hline
Wnr-CE-DL     & 2.50      & 2.17e-05      & Wnr-CE-DR       & 3.31    & 4.63e-05 \\
\hline
Capon         & 5.44      & 4.42e-05      & Capon-DL        & 4.52    & 2.50e-05 \\
\hline
ZF            & 2.12      & 2.54e-05      & Kernel          & 1.07    & 1.60e-04 \\
\hline
Kernel-DL     & \blue{0.80}      & 5.59e-05      &                 &         &      \\
\toprule[0.8pt]
\multicolumn{6}{l}{
\tabincell{l}{
\textbf{Wnr}: The abbreviation for Wiener. \\
\textbf{Time}: The training time averaged on 250 Monte--Carlo episodes.
}
}
\end{tabular}
\end{table}

\begin{table}[!htbp]
\centering
\caption{Experimental Results (Pilot Size = 15)}
\label{tab:scenario-2-2}
\begin{tabular}{lrr|lrr}
\bottomrule[0.8pt]
Combiner  & MSE & Time & Combiner  & MSE & Time \\
\specialrule{0.7pt}{0pt}{0pt}
Wnr           & 1.38      & 1.65e-04      & Wnr-DL          & 1.23    & 1.10e-05 \\
\hline
Wnr-DR        & 1.07      & \red{3.21e+00}      & Wnr-CE          & 1.38    & 4.44e-05 \\
\hline
Wnr-CE-DL     & 1.30      & 2.12e-05      & Wnr-CE-DR       & 1.39    & 4.28e-05 \\
\hline
Capon         & 4.48      & 4.31e-05      & Capon-DL        & 4.34    & 2.42e-05 \\
\hline
ZF            & 2.97      & 2.44e-05      & Kernel          & 1.12    & 1.94e-04 \\
\hline
Kernel-DL     & \blue{0.70}      & 9.23e-05      &                 &         &      \\
\toprule[0.8pt]
\end{tabular}
\end{table}

\begin{table}[!htbp]
\centering
\caption{Experimental Results (Pilot Size = 20)}
\label{tab:scenario-2-3}
\begin{tabular}{lrr|lrr}
\bottomrule[0.8pt]
Combiner  & MSE & Time & Combiner  & MSE & Time \\
\specialrule{0.7pt}{0pt}{0pt}
Wnr           & 1.12      & 1.86e-04      & Wnr-DL          & 1.05    & 1.87e-05 \\
\hline
Wnr-DR        & 0.93      & \red{7.19e+00}      & Wnr-CE          & 1.12    & 5.78e-05 \\
\hline
Wnr-CE-DL     & 1.08      & 3.14e-05      & Wnr-CE-DR       & 1.13    & 6.01e-05 \\
\hline
Capon         & 5.01      & 5.93e-05      & Capon-DL        & 4.94    & 3.81e-05 \\
\hline
ZF            & 3.82      & 3.56e-05      & Kernel          & 1.20    & 4.48e-04 \\
\hline
Kernel-DL     & \blue{0.66}      & 3.11e-04      &                 &         &      \\
\toprule[0.8pt]
\end{tabular}
\end{table}

\begin{table}[!htbp]
\centering
\caption{Experimental Results (Pilot Size = 25)}
\label{tab:scenario-2-4}
\begin{tabular}{lrr|lrr}
\bottomrule[0.8pt]
Combiner  & MSE & Time & Combiner  & MSE & Time \\
\specialrule{0.7pt}{0pt}{0pt}
Wnr           & 0.92      & 1.41e-04      & Wnr-DL          & 0.88    & 1.11e-05 \\
\hline
Wnr-DR        & 0.80      & \red{4.22e+00}      & Wnr-CE          & 0.92    & 5.02e-05 \\
\hline
Wnr-CE-DL     & 0.90      & 2.44e-05      & Wnr-CE-DR       & 0.92    & 4.78e-05 \\
\hline
Capon         & 4.94      & 4.93e-05      & Capon-DL        & 4.89    & 2.85e-05 \\
\hline
ZF            & 4.06      & 2.72e-05      & Kernel          & 1.14    & 4.26e-04 \\
\hline
Kernel-DL     & \blue{0.60}      & 2.95e-04      &                 &         &      \\
\toprule[0.8pt]
\end{tabular}
\end{table}

\begin{table}[!htbp]
\centering
\caption{Experimental Results (Pilot Size = 50)}
\label{tab:scenario-2-5}
\begin{tabular}{lrr|lrr}
\bottomrule[0.8pt]
Combiner  & MSE & Time & Combiner  & MSE & Time \\
\specialrule{0.7pt}{0pt}{0pt}
Wnr           & 0.69      & 1.75e-04      & Wnr-DL          & 0.68    & 1.85e-05 \\
\hline
Wnr-DR        & 0.65      & \red{6.10e+00}      & Wnr-CE          & 0.69    & 5.81e-05 \\
\hline
Wnr-CE-DL     & 0.68      & 3.03e-05      & Wnr-CE-DR       & 0.70    & 5.90e-05 \\
\hline
Capon         & 6.95      & 5.97e-05      & Capon-DL        & 6.93    & 3.75e-05 \\
\hline
ZF            & 6.36      & 3.38e-05      & Kernel          & 0.92    & 1.81e-03 \\
\hline
Kernel-DL     & \blue{0.53}      & 1.67e-03      &                 &         &      \\
\toprule[0.8pt]
\end{tabular}
\end{table}

\begin{table}[!htbp]
\centering
\caption{Experimental Results (Pilot Size = 100)}
\label{tab:scenario-2-6}
\begin{tabular}{lrr|lrr}
\bottomrule[0.8pt]
Combiner  & MSE & Time & Combiner  & MSE & Time \\
\specialrule{0.7pt}{0pt}{0pt}
Wnr           & 0.57      & 3.41e-04      & Wnr-DL          & 0.57    & 3.64e-05 \\
\hline
Wnr-DR        & 0.55      & \red{4.96e+00}      & Wnr-CE          & 0.57    & 6.35e-05 \\
\hline
Wnr-CE-DL     & 0.57      & 2.93e-05      & Wnr-CE-DR       & 0.58    & 6.07e-05 \\
\hline
Capon         & 9.89      & 6.88e-05      & Capon-DL        & 9.88    & 3.99e-05 \\
\hline
ZF            & 9.45      & 3.27e-05      & Kernel          & 0.72    & 5.93e-03 \\
\hline
Kernel-DL     & \blue{0.49}      & 5.83e-03      &                 &         &      \\
\toprule[0.8pt]
\end{tabular}
\end{table}

\section{Conclusions}\label{sec:conclusion}
This article introduces a unified mathematical framework for receive combining of wireless signals from the perspective of data-driven machine learning, which reveals that channel estimation is not a necessary operation. 
To combat the limited pilot size and several types of uncertainties in the signal model, the distributionally robust (DR) receive combining framework is then suggested. We prove that the diagonal-loading (DL) methods are distributionally robust against the scarcity of pilot data and the uncertainties in the signal model. In addition, we generalize the diagonal-loading methods to achieve better estimation performance (e.g., the DR Wiener beamformer using $F$-norm for uncertainty quantification), at the cost of significantly higher computational burdens. Experiments suggest that nonlinear combiners such as the Kernel and the Kernel-DL methods have the potential when the pilot size is small and/or the signal model is not linear Gaussian. Compared with the Kernel and the Kernel-DL combiners, neural-network-based solutions \cite{ye2017power,van2022transfer} have a stronger expressive capability of nonlinearities, which however are unscalable in the numbers of transmit and receive antennas, and significantly more time-consuming in training and more troublesome in tuning hyper-parameters (e.g., the number of layers and the number of neurons in each layer) than the studied eleven combiners.

\appendices

\section{Structured Representation of Nonlinear Functions}\label{append:func-representation}
In Section \ref{sec:prelimilary}, we have reviewed two popular frameworks for representing (nonlinear) functions: reproducing kernel Hilbert spaces (RKHS) and neural network function spaces (NNFS). Typical kernel functions $\ker(\cdot, \cdot)$ to define RKHSs include Gaussian kernel, Matern kernel, Linear kernel, Laplacian kernel, and Polynomial kernel.
Mathematical details of these kernel functions can be found in \cite[Subsec.~14.2]{murphy2012machine}, \cite[Ex.~1]{shafieezadeh2019regularization}.
Typical activation functions $\sigma(\cdot)$ to define NNFSs include Hyperbolic tangent (i.e, tanh) function, Softmax function, Sigmoid function, Rectified linear unit (ReLU) function, and Exponential linear unit (ELU) function. 
Mathematical details of these activation functions can be found in \cite[Ex.~2]{shafieezadeh2019regularization}.

\section{Details on Real-Space Signal Representation}\label{append:details-real-representation}
Let $\mat R_x \defeq \E \rvec x \rvec x^\H$, $\mat C_x \defeq \E \rvec x \rvec x^\T$, $\mat C_s \defeq \E \rvec s \rvec s^\T$, and $\mat C_v \defeq \E \rvec v \rvec v^\T = \mat 0$. 
We have
\[
\begin{array}{cl}
\mat R_{\ul x} \defeq \E{\underline{\rvec x} \underline{\rvec x}^\T} = 
\displaystyle \frac{1}{2}\left[
\begin{array}{cc}
   \Re (\mat R_x + \mat C_x)  & \Im (-\mat R_x + \mat C_x)\\
   \Im (\mat R_x + \mat C_x)  & \Re (\mat R_x - \mat C_x)
\end{array}
\right].
\end{array}
\]
\[
\begin{array}{cl}
\mat R_{\ul s} \defeq \E{\underline{\rvec s} \underline{\rvec s}^\T} &= 
\displaystyle \frac{1}{2}\left[
\begin{array}{cc}
   \Re (\mat R_s + \mat C_s)  & \Im (-\mat R_s + \mat C_s)\\
   \Im (\mat R_s + \mat C_s)  & \Re (\mat R_s - \mat C_s)
\end{array}
\right],
\end{array}
\]
and
\[
\begin{array}{cl}
\mat R_{\ul v} \defeq \E{\underline{\rvec v} \underline{\rvec v}^\T} &= \displaystyle \frac{1}{2}\left[
\begin{array}{cc}
   \Re \mat R_v  & \Im -\mat R_v\\
   \Im \mat R_v & \Re \mat R_v
\end{array}
\right].
\end{array}
\]
Note that the following identities hold:
$
    \mat R_x = \mat H \mat R_s \mat H^\H + \mat R_v
$, 
$
    \mat C_x = \mat H \mat C_s \mat H^\T
$, 
$
    \mat R_{\ul x} = \matuul H \cdot \mat R_{\ul s} \cdot \matuul H^\T + \mat R_{\ul v}
$, 
and
$
    \mat R_{\ul{xs}} = \matuul H \cdot \mat R_{\ul s}
$.

\section{Extensive Reading on Distributional Uncertainty}\label{append:dist-uncertainty}
\subsection{Generalization Error and Distributional Robustness}
We use \eqref{eq:opt-esti} and \eqref{eq:esti-erm} as examples to illustrate the concepts. Supposing that $\vec \phi^\star$ solves the true problem \eqref{eq:opt-esti} and $\vec \phi^\star_{\text{ERM}}$ solves the surrogate problem \eqref{eq:esti-erm}, we have
\begin{equation}\label{eq:gen-err}
\begin{array}{l}
    \displaystyle \min_{\vec \phi} \Tr \E_{(\rvec x, \rvec s) \sim \P_{\rvec x, \rvec s}}[\vec \phi(\rvec x) - \rvec s][\vec \phi(\rvec x) - \rvec s]^\H \\
    
    \quad \quad = \Tr \E_{(\rvec x, \rvec s) \sim \P_{\rvec x, \rvec s}}[\vec \phi^\star(\rvec x) - \rvec s][\vec \phi^\star(\rvec x) - \rvec s]^\H \\

    \quad \quad \le \Tr \E_{(\rvec x, \rvec s) \sim \P_{\rvec x, \rvec s}}[\vec \phi^\star_{\text{ERM}}(\rvec x) - \rvec s][\vec \phi^\star_{\text{ERM}}(\rvec x) - \rvec s]^\H.
\end{array}
\end{equation}
To clarify further, the testing error in the last line (evaluated at the true distribution $\P_{\rvec x, \rvec s}$) of the learned estimator $\vec \phi^\star_{\text{ERM}}$ may be (much) larger than the optimal error in the first two lines, although $\vec \phi^\star_{\text{ERM}}$ has the smallest training error (evaluated at the nominal distribution $\Ph_{\rvec x, \rvec s}$), i.e.,
\begin{equation}\label{eq:train-err}
\begin{array}{l}
    \displaystyle \min_{\vec \phi}  \Tr \E_{(\rvec x, \rvec s) \sim \Ph_{\rvec x, \rvec s}}[\vec \phi(\rvec x) - \rvec s][\vec \phi(\rvec x) - \rvec s]^\H \\

    \quad \quad = \min_{\vec \phi}  \Tr \frac{1}{L} \sum^L_{i=1}[\vec \phi(\vec x_i) - \vec s_i][\vec \phi(\vec x_i) - \vec s_i]^\H \\
    
    \quad \quad = \Tr \frac{1}{L} \sum^L_{i=1}[\vec \phi^\star_{\text{ERM}}(\vec x_i) - \vec s_i][\vec \phi^\star_{\text{ERM}}(\vec x_i) - \vec s_i]^\H \\

    \quad \quad \le \Tr \frac{1}{L} \sum^L_{i=1}[\vec \phi^\star(\vec x_i) - \vec s_i][\vec \phi^\star(\vec x_i) - \vec s_i]^\H.
\end{array}
\end{equation}
In the terminologies of machine learning, the difference between the testing error and the training error, i.e.,
\[
\begin{array}{l}
\Tr \E_{(\rvec x, \rvec s) \sim \P_{\rvec x, \rvec s}}[\vec \phi^\star_{\text{ERM}}(\rvec x) - \rvec s][\vec \phi^\star_{\text{ERM}}(\rvec x) - \rvec s]^\H - \\
\quad \quad \quad \displaystyle \Tr \E_{(\rvec x, \rvec s) \sim \Ph_{\rvec x, \rvec s}}[\vec \phi^\star_{\text{ERM}}(\rvec x) - \rvec s][\vec \phi^\star_{\text{ERM}}(\rvec x) - \rvec s]^\H \\
= 
\Tr \E_{\rvec x, \rvec s}[\vec \phi^\star_{\text{ERM}}(\rvec x) - \rvec s][\vec \phi^\star_{\text{ERM}}(\rvec x) - \rvec s]^\H - \\
\quad \quad \quad \displaystyle \Tr \frac{1}{L} \sum^L_{i=1}[\vec \phi^\star_{\text{ERM}}(\vec x_i) - \vec s_i][\vec \phi^\star_{\text{ERM}}(\vec x_i) - \vec s_i]^\H
\end{array}
\]
is called the \textit{generalization error} of $\vec \phi^\star_{\text{ERM}}$; the difference between the testing error and the optimal error, i.e.,
\[
\begin{array}{l}
\Tr \E_{\rvec x, \rvec s}[\vec \phi^\star_{\text{ERM}}(\rvec x) - \rvec s][\vec \phi^\star_{\text{ERM}}(\rvec x) - \rvec s]^\H - \\
\quad \quad \quad \displaystyle \Tr \E_{\rvec x, \rvec s}[\vec \phi^\star(\rvec x) - \rvec s][\vec \phi^\star(\rvec x) - \rvec s]^\H
\end{array}
\]
is called the \textit{excess risk} of $\vec \phi^\star_{\text{ERM}}$. In machine learning practice, we want to reduce both the generalization error and the excess risk. Most attention in the literature has been particularly paid to reducing generalization errors. Specifically, an upper bound of the true cost $\Tr \E_{(\rvec x, \rvec s) \sim \P_{\rvec x, \rvec s}}[\vec \phi(\rvec x) - \rvec s][\vec \phi(\rvec x) - \rvec s]^\H$ is first found and then minimize the upper bound: by minimizing the upper bound, the true cost can also be reduced. 

\begin{fact}\label{fact:DRO-upper-bound}
Suppose that the true distribution $\P_{0, \rvec x, \rvec s}$ of $(\rvec x, \rvec s)$ is included in $\cal U_{\rvec x, \rvec s}$; for notational clarity, we hereafter distinguish
$\P_{0, \rvec x, \rvec s}$ from $\P_{\rvec x, \rvec s}$. The true objective function evaluated at $\P_{0, \rvec x, \rvec s}$, i.e.,
\begin{equation}\label{eq:true-err}
\Tr \E_{(\rvec x, \rvec s) \sim \P_{0, \rvec x, \rvec s}}[\vec \phi(\rvec x) - \rvec s][\vec \phi(\rvec x) - \rvec s]^\H,~~~\forall \vec \phi \in \cal B,
\end{equation}
is upper bounded by the worst-case objective function of \eqref{eq:opt-esti-dist-robust}, i.e.,
\begin{equation}\label{eq:upper-bnd}
\max_{\P_{\rvec x, \rvec s} \in \cal U_{\rvec x, \rvec s}} \Tr \E_{(\rvec x, \rvec s) \sim \P_{\rvec x, \rvec s}}[\vec \phi(\rvec x) - \rvec s][\vec \phi(\rvec x) - \rvec s]^\H,~~~\forall \vec \phi \in \cal B.
\end{equation}
Therefore, by diminishing the upper bound in \eqref{eq:upper-bnd}, the true estimation error evaluated at $\P_{0, \rvec x, \rvec s}$ can also be reduced. However, the conventional empirical estimation error evaluated at $\Ph_{\rvec x, \rvec s}$ cannot upper bound the true estimation error \eqref{eq:true-err}. This performance guarantee is the benefit of considering the distributionally robust method \eqref{eq:opt-esti-dist-robust}. Due to the weak convergence property of the empirical distribution to the true data-generating distribution, that is, $d(\P_{0, \rvec x, \rvec s},~\Ph_{\rvec x, \rvec s}) \to 0$ as the sample size $L \to \infty$, there exists $\epsilon$ in \eqref{eq:wasserstein-ball}  for every $L$, such that $\P_{0, \rvec x, \rvec s}$ is included in $\cal U_{\rvec x, \rvec s}$ in $\P^L_{0, \rvec x, \rvec s}$-probability ($L$-fold product measure of $\P_{0, \rvec x, \rvec s}$).
\stp
\end{fact}

\subsection{Non-Stationary Channel Statistics}
In the main body of the article (see also Fact \ref{fact:DRO-upper-bound}), we assume that the true data-generating distribution $\P_{0, \rvec x, \rvec s}$ is time-invariant within a frame. In real-world operations, however, this assumption might be untenable.

As shown in Fig. \ref{fig-time-varying-true-dist}, the frame contains eight data units; we suppose that the first four units are pilot symbols and the rest four units are communication-data symbols.

\begin{figure}[!htbp]
    \centering
    \includegraphics[height=1.5cm]{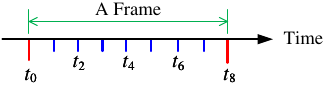}
    \caption{True data-generating distributions might be time-varying in a frame.}
    \label{fig-time-varying-true-dist}
\end{figure}

Let $\P_{0, \rvec x, \rvec s, i}$ denote the true data-generating distribution at time point $t_i$ where $i = 1, 2, \ldots, 8$. Specifically, we have $(\rvec x_i, \rvec s_i) \sim \P_{0, \rvec x, \rvec s, i}$ for every $i$. Therefore, the pilot data set (i.e., the training data set) $\{(\vec x_1, \vec s_1),(\vec x_2, \vec s_2),(\vec x_3, \vec s_3),(\vec x_4, \vec s_4)\}$ can be seen as realizations of the mean distribution $\P_{\text{train}, 0, \rvec x, \rvec s}$ of underlying true training-data distributions where 
\(
\P_{\text{train}, 0, \rvec x, \rvec s} = \sum^4_{i = 1} h_i \P_{0, \rvec x, \rvec s, i}
\), 
which is a mixture distribution with mixing weights $0 \le h_1, h_2, h_3, h_4 \le 1$; $\sum^4_{i = 1} h_i = 1$. Similarly, the communication data set (i.e., the testing data set) $\{(\vec x_5, \vec s_5),(\vec x_6, \vec s_6),(\vec x_7, \vec s_7),(\vec x_8, \vec s_8)\}$ can be seen as realizations of the mean $\P_{\text{test}, 0, \rvec x, \rvec s}$ of the underlying true testing-data distributions where
\(
\P_{\text{test}, 0, \rvec x, \rvec s} = \sum^8_{i = 5} h_i \P_{0, \rvec x, \rvec s, i}
\), 
with mixing weights $0 \le h_5, h_6, h_7, h_8 \le 1$; $\sum^8_{i = 5} h_i = 1$.

Suppose that 
\[
d(\Ph_{\text{train}, \rvec x, \rvec s},~\P_{\text{train}, 0, \rvec x, \rvec s}) \le \epsilon_1,
\] 
where $\Ph_{\text{train}, \rvec x, \rvec s} \defeq \frac{1}{4} \sum^4_{i=1} \delta_{(\vec x_i, \vec s_i)}$ is the data-driven estimate of $\P_{\text{train}, 0, \rvec x, \rvec s}$ and 
\[
d(\P_{\text{train}, 0, \rvec x, \rvec s},~\P_{\text{test}, 0, \rvec x, \rvec s}) \le \epsilon_2,
\]
for some $\epsilon_1, \epsilon_2 \ge 0$. We have the uncertainty quantification 
\[
d(\P_{\text{test}, 0, \rvec x, \rvec s},~\Ph_{\text{train}, \rvec x, \rvec s}) \le \epsilon \defeq \epsilon_1 + \epsilon_2.
\]
Therefore, the distributionally robust modeling and solution framework is still valid to hedge against the distributional uncertainty in the nominal distribution $\Ph_{\text{train}, \rvec x, \rvec s}$ compared to the underlying true distribution $\P_{\text{test}, 0, \rvec x, \rvec s}$. When $\P_{\text{train}, 0, \rvec x, \rvec s} = \P_{\text{test}, 0, \rvec x, \rvec s}$, as assumed in the main body of the article, we have $\epsilon_1 \to 0$ and $\epsilon \to \epsilon_2 = 0$ as the pilot size tends to infinity; however, when $\P_{\text{train}, 0, \rvec x, \rvec s} \neq \P_{\text{test},0, \rvec x, \rvec s}$, the radius $\epsilon \to \epsilon_2 \neq 0$ although $\epsilon_1 \to 0$.

Another justification for the DRO method is as follows. Suppose that there exists $\epsilon \ge 0$ such that
\[
d(\P_{0, \rvec x, \rvec s, i},~\Ph_{\text{train}, \rvec x, \rvec s}) \le \epsilon,~~~\forall i \in \{1, 2, \ldots, 8\}.
\] 
It means that, at every snapshot in the frame, the true data-generating distribution is included in the uncertainty set. Hence, the DRO cost can still upper bound the true cost even though the true distribution is time-varying; cf. Fact \ref{fact:DRO-upper-bound}.

\section{Additional Discussions on Distributionally Robust Estimation}\label{append:DRO-supplementary}
To develop this article, the typical minimum mean-squared error (MSE) criterion is employed; see \eqref{eq:opt-esti} and \eqref{eq:opt-BF}. Accordingly, the distributionally robust receive combining framework in this article is exemplified using the MSE cost function. The cost function for wireless signal estimation, however, can be any Borel-measurable function $h: \C^M \times \C^M \to \R_+$. As a result, the optimal estimation problem under the distribution $\P_{\rvec x, \rvec s}$ is given by
\begin{equation}\label{eq:opt-esti-cost-h}
    \min_{\vec \phi \in \cal B_{\C^{N} \to \C^{M}}} \E_{\rvec x, \rvec s} h[\vec \phi(\rvec x), \rvec s].
\end{equation}
Specific examples of $h$ in wireless communications can be, e.g., mean absolute error, Huber's cost function \cite{wang1998robust,katkovnik2006performance} where $h$ is no longer quadratic as in \eqref{eq:opt-esti} and \eqref{eq:opt-BF}. Accordingly, when the distributional uncertainty exists in $\P_{\rvec x, \rvec s}$, the distributionally robust receive combining framework becomes
\begin{equation}\label{eq:opt-esti-dist-robust-cost-h}
    \min_{\vec \phi \in \cal B_{\C^{N} \to \C^{M}}} \max_{\P_{\rvec x, \rvec s} \in \cal U_{\rvec x, \rvec s}} \E_{\rvec x, \rvec s} h[\vec \phi(\rvec x), \rvec s].
\end{equation}
Problem \eqref{eq:opt-esti-dist-robust-cost-h} is generally challenging to solve because it is an infinite-dimensional program. Therefore, in practice, we can limit the feasible region of $\vec \phi$ to a parameterized subspace of $\cal B_{\C^{N} \to \C^{M}}$, for example, a reproducing kernel Hilbert space $\cal H$ or a neural network function space $\cal K$; see Section \ref{sec:prelimilary}. Consequently, Problem \eqref{eq:opt-esti-dist-robust-cost-h} is approximated by the following finite-dimensional (in terms of $\mat W$) program
\begin{equation}\label{eq:opt-esti-dist-robust-cost-h-approx}
    \min_{\mat W} \max_{\P_{\rvec x, \rvec s} \in \cal U_{\rvec x, \rvec s}} \E_{\rvec x, \rvec s} h[\vec \phi_{\mat W}(\rvec x), \rvec s],
\end{equation}
where $\mat W$ parameterizes $\vec \phi$ and lies in real or complex coordinate spaces; note that both $\cal H$ and $\cal K$ can be dense in $\cal B$. 
Under the MSE cost function, \eqref{eq:opt-esti-dist-robust-cost-h-approx} is particularized in \eqref{eq:BF-dist-robust} for linear function spaces, in \eqref{eq:opt-esti-dist-robust-RKHS} for reproducing kernel Hilbert spaces, and in \eqref{eq:opt-esti-dist-robust-NN} for neural network function spaces, which build this article in a technically tractable manner.

The distributionally robust receive combining problem \eqref{eq:opt-esti-dist-robust-cost-h} under generic cost functions $h$ and generic feasible regions of $\vec \phi$ can be technically challenging. Even for the simplified problem \eqref{eq:opt-esti-dist-robust-cost-h-approx}, the solution method can be quite complex, and closed-form solutions cannot be generally guaranteed; see, e.g., \cite{rahimian2022frameworks,kuhn2024distributionally}. The complication further arises when the distributional uncertainty sets $\cal U_{\rvec x, \rvec s}$ for $\P_{\rvec x, \rvec s}$ are complicated; see, e.g., \cite{kuhn2019wasserstein}. Therefore, this article serves as the starting point of distributionally robust receive combining, in which \bfit{closed-form solutions} are largely ensured by leveraging
\begin{enumerate}[\hspace{1em}F1)]
    \item the MSE cost function as in \eqref{eq:opt-esti} and \eqref{eq:opt-BF};

    \item the linear function spaces as in \eqref{eq:BF-dist-robust} and reproducing kernel Hilbert spaces as in \eqref{eq:opt-esti-dist-robust-RKHS};

    \item the second-moment-based uncertainty sets in Definitions \ref{def:uncertainty-set-additive-moment}, \ref{def:uncertainty-set-diag-loading}, \ref{def:generalized-uncertainty-set-diag-loading}, and \ref{def:uncertainty-set-multiplicative-moment}; see also Corollary \ref{cor:solutions-under-moments}, Claim \ref{claim:opt-estimator-RKHS}, and Example \ref{exam:opt-estimator-RKHS}.
\end{enumerate}

Note that even under the features F1) and F2), the closed-form solutions cannot be guaranteed. For example, if Wasserstein or F-norm uncertainty sets are used, the associated distributionally robust receive combining problems can be computationally heavy; see \eqref{eq:BF-dist-robust-max-wasserstein} and \eqref{eq:BF-dist-robust-max-F-norm} as well as Propositions \ref{prop:wasserstein-R-x} and \ref{prop:F-norm-R-x}. However, for emerging high-performance computing devices, the computational burden may be no longer an issue in the future. Hence, advanced distributionally robust receive combining formulations based on \eqref{eq:opt-esti-dist-robust-cost-h} and \eqref{eq:opt-esti-dist-robust-cost-h-approx} are still attractive for future-generation communication systems. This article seeks to provide a foundation for this direction.

\section{Proof of Lemma \ref{lem:lin-BF-dist-robust-dual}}\label{append:dual}
\begin{proof}
The objective function of Problem \eqref{eq:BF-dist-robust} equals to
\begin{equation}\label{BF-dist-robust-linobj}
\left\langle
        \left[
        \begin{array}{cc}
           \mat W^\H \mat W  &  -\mat W^\H \\
           -\mat W  &  \mat I_M
        \end{array}
        \right]
        ,~
        \left[
        \begin{array}{cc}
           \mat R_x  &  \mat R_{xs} \\
           \mat R^\H_{xs}  &  \mat R_s
        \end{array}
        \right]
\right\rangle,
\end{equation}
where $\langle \mat A, \mat B\rangle \defeq \Tr \mat A^\H \mat B$ for two matrices $\mat A$ and $\mat B$. Therefore, the objective function of \eqref{eq:BF-dist-robust} is convex in $\mat W$ and linear (thus concave) in the matrix variable $\mat R$. Hence, due to Sion's minimax theorem \cite[Corollary~3.3]{sion1958general}, Problem \eqref{eq:BF-dist-robust} is equivalent to 
\begin{equation}\label{eq:BF-dist-robust-dual}
    \begin{array}{cl}
        \displaystyle \max_{\mat R} \min_{\mat W} &\Tr \big[\mat W \mat R_{x} \mat W^\H - \mat W \mat R_{xs} - \mat R^\H_{xs} \mat W^\H + \mat R_{s}\big] \\
        \st & d_0(\mat R,~\math R) \le \epsilon_0, \\

        & \mat R \succeq \mat 0.
    \end{array}
\end{equation}
Note that the feasible region of $\mat R$ is compact convex, and that of $\mat W$ (i.e., $\C^{M \times N}$) is convex.

For every given $\mat R$, the inner minimization sub-problem of \eqref{eq:BF-dist-robust-dual} is solved by the Wiener beamformer $\mat W^\star_{\text{Wiener}} = \mat R^{\H}_{xs}\mat R^{-1}_{x}$, which transforms \eqref{eq:BF-dist-robust-dual} to \eqref{eq:BF-dist-robust-max}. This completes the proof.
\stp
\end{proof}

\section{Proof of Theorem \ref{thm:f-increasing-R-x}}\label{append:f-increasing-R-x}
\begin{proof}
Consider the following optimization problem
\begin{equation}\label{eq:f1-opt}
    \begin{array}{cl}
        \displaystyle \max_{\mat R} &\Tr \big[ -\mat R^{\H}_{xs}\mat R^{-1}_{x} \mat R_{xs} + \mat R_{s}\big] \\
        \st & \mat R \succeq \mat R_2, \\
        & \mat R_x \succ \mat 0,
    \end{array}
\end{equation}
which, due to Lemma \ref{lem:lin-BF-dist-robust-dual}, is equivalent [in the sense of the same optimal objective value and maximizer(s) $\mat R^\star$] to
\begin{equation}\label{eq:f1-transform}
    \begin{array}{cl}
        \displaystyle \min_{\mat W} \max_{\mat R} & 
        \left\langle
        \left[
        \begin{array}{cc}
           \mat W^\H \mat W  &  -\mat W^\H \\
           -\mat W  &  \mat I_M
        \end{array}
        \right]
        ,~
        \left[
        \begin{array}{cc}
           \mat R_x  &  \mat R_{xs} \\
           \mat R^\H_{xs}  &  \mat R_s
        \end{array}
        \right]
        \right\rangle \\
        \st & \mat R \succeq \mat R_2, \\
        & \mat R_x \succ \mat 0.
    \end{array}
\end{equation}
Note that 
$
    \left[
        \begin{array}{cc}
           \mat W^\H \mat W  &  -\mat W^\H \\
           -\mat W  &  \mat I_M
        \end{array}
    \right] \succeq \mat 0,
$ 
because for all $\vec x \in \C^{N}$ and $\vec y \in \C^{M}$, we have
\[
[\vec x^\H,~\vec y^\H]
\left[
        \begin{array}{cc}
           \mat W^\H \mat W  &  -\mat W^\H \\
           -\mat W  &  \mat I_M
        \end{array}
\right]
\left[
        \begin{array}{c}
           \vec x \\
           \vec y
        \end{array}
\right] = \|\mat W \vec x - \vec y\|^2_2 
\ge 0.
\]
Therefore, for every given $\mat W$, the objective function of \eqref{eq:f1-transform} is increasing in $\mat R$. As a result, the objective value of \eqref{eq:f1-opt} is lower-bounded at $\mat R_2$: To be specific, $\forall \mat R \succeq \mat R_2$, we have
\[
\Tr \big[ -\mat R^{\H}_{xs}\mat R^{-1}_{x} \mat R_{xs} + \mat R_{s}\big] \\
\ge \Tr \big[ -\mat R^{\H}_{2,xs}\mat R^{-1}_{2,x} \mat R_{2,xs} + \mat R_{2,s}\big],
\]
i.e., $f_1(\mat R) \ge f_1(\mat R_2)$, which proves the first part.

On the other hand, if $\mat R_{1, x} \succeq \mat R_{2, x} \succ \mat 0$, we have $\mat R^{-1}_{2, x} \succeq \mat R^{-1}_{1, x}$. As a result,
$
    f_2(\mat R_{1, x}) - f_2(\mat R_{2, x}) = \tr{\mat R^\H_{xs} (\mat R^{-1}_{2, x} - \mat R^{-1}_{1, x}) \mat R_{xs}} \ge 0
$, completing the proof.
 \stp
\end{proof}

\section{Proof of Proposition \ref{prop:RHKS-reformulation}}\label{append:BF-dist-robust-z}
\begin{proof}
Letting $\rvecul z \defeq \vec \varphi(\rvecul x)$, \eqref{eq:opt-esti-dist-robust-RKHS} can be rewritten as
\begin{equation}\label{eq:opt-esti-dist-robust-RKHS-z}
    \min_{\mat W \in \R^{2M \times L}} \max_{\P_{\rvecul z, \rvecul s} \in \cal U_{\rvecul z, \rvecul s}} \Tr \E_{\rvecul z, \rvecul s}[\mat W \rvecul z - \rvecul s][\mat W \rvecul z - \rvecul s]^\T.
\end{equation}
Tantamount to the distributionally robust beamforming problem \eqref{eq:BF-dist-robust}, Problem \eqref{eq:opt-esti-dist-robust-RKHS-z} reduces to \eqref{eq:BF-dist-robust-z} 
where 
\[
   \math {R}_{\ul z} \defeq \frac{1}{L} \sum^L_{i = 1} \vecul z_i \vecul z^\T_i = \frac{1}{L} \sum^L_{i = 1} \vec \varphi(\vecul x_i) \vec \varphi^\T(\vecul x_i) = \frac{1}{L} \mat K^2,
\]
\[
   \math {R}_{\ul{zs}} \defeq \frac{1}{L} \sum^L_{i = 1} \vecul z_i \vecul s^\T_i = \frac{1}{L} \sum^L_{i = 1} \vec \varphi(\vecul x_i) \cdot \vecul s^\T_i = \frac{1}{L} \mat K \matul S^\T,
\]
\[
   \math {R}_{\ul s} \defeq \frac{1}{L} \sum^L_{i = 1} \vecul s_i \vecul s^\T_i = \frac{1}{L} \sum^L_{i = 1}\vecul s_i \cdot \vecul s^\T_i = \frac{1}{L} \matul S \matul S^\T,
\]
and
\[
    \mat K \defeq [\vec \varphi(\vecul x_1), \vec \varphi(\vecul x_2), \ldots, \vec \varphi(\vecul x_L)] \in \R^{L \times L}.
\]
The rest claims are due to Lemma \ref{lem:lin-BF-dist-robust-dual}; NB: $\mat K$ is invertible.
\stp
\end{proof}

\bibliographystyle{IEEEtran}
\bibliography{References}

\begin{thebibliography}{10}
\providecommand{\url}[1]{#1}
\csname url@samestyle\endcsname
\providecommand{\newblock}{\relax}
\providecommand{\bibinfo}[2]{#2}
\providecommand{\BIBentrySTDinterwordspacing}{\spaceskip=0pt\relax}
\providecommand{\BIBentryALTinterwordstretchfactor}{4}
\providecommand{\BIBentryALTinterwordspacing}{\spaceskip=\fontdimen2\font plus
\BIBentryALTinterwordstretchfactor\fontdimen3\font minus
  \fontdimen4\font\relax}
\providecommand{\BIBforeignlanguage}[2]{{%
\expandafter\ifx\csname l@#1\endcsname\relax
\typeout{** WARNING: IEEEtran.bst: No hyphenation pattern has been}%
\typeout{** loaded for the language `#1'. Using the pattern for}%
\typeout{** the default language instead.}%
\else
\language=\csname l@#1\endcsname
\fi
#2}}
\providecommand{\BIBdecl}{\relax}
\BIBdecl

\bibitem{lo1991nonlinear}
T.~Lo, H.~Leung, and J.~Litva, ``Nonlinear beamforming,'' \emph{Electronics
  Letters}, vol.~4, no.~27, pp. 350--352, 1991.

\bibitem{yang2015fifty}
S.~Yang and L.~Hanzo, ``Fifty years of {MIMO} detection: The road to
  large-scale {MIMOs},'' \emph{IEEE Commun. Surveys Tuts.}, vol.~17, no.~4, pp.
  1941--1988, 2015.

\bibitem{elbir2023twenty}
A.~M. Elbir, K.~V. Mishra, S.~A. Vorobyov, and R.~W. Heath, ``Twenty-five years
  of advances in beamforming: From convex and nonconvex optimization to
  learning techniques,'' \emph{IEEE Signal Processing Mag.}, vol.~40, no.~4,
  pp. 118--131, 2023.

\bibitem{chen2008adaptive}
S.~Chen, S.~Tan, L.~Xu, and L.~Hanzo, ``Adaptive minimum error-rate filtering
  design: A review,'' \emph{Signal Processing}, vol.~88, no.~7, pp. 1671--1697,
  2008.

\bibitem{chen2008symmetric}
S.~Chen, A.~Wolfgang, C.~J. Harris, and L.~Hanzo, ``Symmetric {RBF} classifier
  for nonlinear detection in multiple-antenna-aided systems,'' \emph{IEEE
  Trans. Neural Networks}, vol.~19, no.~5, pp. 737--745, 2008.

\bibitem{navia2010approximate}
A.~Navia-Vazquez, M.~Martinez-Ramon, L.~E. Garcia-Munoz, and C.~G.
  Christodoulou, ``Approximate kernel orthogonalization for antenna array
  processing,'' \emph{IEEE Trans. Antennas Propagat.}, vol.~58, no.~12, pp.
  3942--3950, 2010.

\bibitem{neinavaie2020lossless}
M.~Neinavaie, M.~Derakhtian, and S.~A. Vorobyov, ``Lossless dimension reduction
  for integer least squares with application to sphere decoding,'' \emph{IEEE
  Trans. Signal Processing}, vol.~68, pp. 6547--6561, 2020.

\bibitem{liao2023deep}
J.~Liao, J.~Zhao, F.~Gao, and G.~Y. Li, ``Deep learning aided low complex
  breadth-first tree search for {MIMO} detection,'' \emph{IEEE Trans. Wireless
  Commun.}, 2023.

\bibitem{awan2023robust}
D.~A. Awan, R.~L. Cavalcante, M.~Yukawa, and S.~Stanczak, ``Robust online
  multiuser detection: A hybrid model-data driven approach,'' \emph{IEEE Trans.
  Signal Processing}, 2023.

\bibitem{ye2017power}
H.~Ye, G.~Y. Li, and B.-H. Juang, ``Power of deep learning for channel
  estimation and signal detection in {OFDM} systems,'' \emph{IEEE Wireless
  Commun. Lett.}, vol.~7, no.~1, pp. 114--117, 2017.

\bibitem{he2020model}
H.~He, C.-K. Wen, S.~Jin, and G.~Y. Li, ``Model-driven deep learning for {MIMO}
  detection,'' \emph{IEEE Trans. Signal Processing}, vol.~68, pp. 1702--1715,
  2020.

\bibitem{van2022transfer}
N.~Van~Huynh and G.~Y. Li, ``Transfer learning for signal detection in wireless
  networks,'' \emph{IEEE Wireless Commun. Lett.}, vol.~11, no.~11, pp.
  2325--2329, 2022.

\bibitem{li2003robust}
J.~Li, P.~Stoica, and Z.~Wang, ``On robust {Capon} beamforming and diagonal
  loading,'' \emph{IEEE Trans. Signal Processing}, vol.~51, no.~7, pp.
  1702--1715, 2003.

\bibitem{lorenz2005robust}
R.~G. Lorenz and S.~P. Boyd, ``Robust minimum variance beamforming,''
  \emph{IEEE Trans. Signal Processing}, vol.~53, no.~5, pp. 1684--1696, 2005.

\bibitem{zhang2015robust}
X.~Zhang, Y.~Li, N.~Ge, and J.~Lu, ``Robust minimum variance beamforming under
  distributional uncertainty,'' in \emph{2015 IEEE International Conference on
  Acoustics, Speech and Signal Processing (ICASSP)}.\hskip 1em plus 0.5em minus
  0.4em\relax IEEE, 2015, pp. 2514--2518.

\bibitem{li2017distributionally}
B.~Li, Y.~Rong, J.~Sun, and K.~L. Teo, ``A distributionally robust minimum
  variance beamformer design,'' \emph{IEEE Signal Processing Lett.}, vol.~25,
  no.~1, pp. 105--109, 2017.

\bibitem{huang2022robust}
Y.~Huang, W.~Yang, and S.~A. Vorobyov, ``Robust adaptive beamforming maximizing
  the worst-case {SINR} over distributional uncertainty sets for random inc
  matrix and signal steering vector,'' in \emph{ICASSP 2022-2022 IEEE
  International Conference on Acoustics, Speech and Signal Processing
  (ICASSP)}.\hskip 1em plus 0.5em minus 0.4em\relax IEEE, 2022, pp. 4918--4922.

\bibitem{huang2023robust}
Y.~Huang, H.~Fu, S.~A. Vorobyov, and Z.-Q. Luo, ``Robust adaptive beamforming
  via worst-case {SINR} maximization with nonconvex uncertainty sets,''
  \emph{IEEE Trans. Signal Processing}, vol.~71, pp. 218--232, 2023.

\bibitem{cox1987robust}
H.~Cox, R.~Zeskind, and M.~Owen, ``Robust adaptive beamforming,'' \emph{IEEE
  Trans. Acoust., Speech, Signal Processing}, vol.~35, no.~10, pp. 1365--1376,
  1987.

\bibitem{harmanci2000relationships}
K.~Harmanci, J.~Tabrikian, and J.~L. Krolik, ``Relationships between adaptive
  minimum variance beamforming and optimal source localization,'' \emph{IEEE
  Trans. Signal Processing}, vol.~48, no.~1, pp. 1--12, 2000.

\bibitem{liu2018toward}
F.~Liu, L.~Zhou, C.~Masouros, A.~Li, W.~Luo, and A.~Petropulu, ``Toward
  dual-functional radar-communication systems: Optimal waveform design,''
  \emph{IEEE Trans. Signal Processing}, vol.~66, no.~16, pp. 4264--4279, 2018.

\bibitem{zhang2021overview}
J.~A. Zhang, F.~Liu, C.~Masouros, R.~W. Heath, Z.~Feng, L.~Zheng, and
  A.~Petropulu, ``An overview of signal processing techniques for joint
  communication and radar sensing,'' \emph{IEEE J. Select. Topics Signal
  Processing}, vol.~15, no.~6, pp. 1295--1315, 2021.

\bibitem{xiong2023fundamental}
Y.~Xiong, F.~Liu, Y.~Cui, W.~Yuan, T.~X. Han, and G.~Caire, ``On the
  fundamental tradeoff of integrated sensing and communications under
  {Gaussian} channels,'' \emph{IEEE Trans. Inform. Theory}, 2023.

\bibitem{murphy2012machine}
K.~P. Murphy, \emph{Machine Learning: A Probabilistic Perspective}.\hskip 1em
  plus 0.5em minus 0.4em\relax MIT Press, 2012.

\bibitem{bishop2006pattern}
C.~M. Bishop and N.~M. Nasrabadi, \emph{Pattern Recognition and Machine
  Learning}.\hskip 1em plus 0.5em minus 0.4em\relax Springer, 2006, vol.~4,
  no.~4.

\bibitem{li2023towards}
G.~Li and J.~Ding, ``Towards understanding variation-constrained deep neural
  networks,'' \emph{IEEE Trans. Signal Processing}, vol.~71, pp. 631--640,
  2023.

\bibitem{shafieezadeh2019regularization}
S.~Shafieezadeh-Abadeh, D.~Kuhn, and P.~M. Esfahani, ``Regularization via mass
  transportation,'' \emph{Journal of Machine Learning Research}, vol.~20, no.
  103, pp. 1--68, 2019.

\bibitem{staib2019distributionally}
M.~Staib and S.~Jegelka, ``Distributionally robust optimization and
  generalization in kernel methods,'' \emph{Advances in Neural Information
  Processing Systems}, vol.~32, 2019.

\bibitem{delage2010distributionally}
E.~Delage and Y.~Ye, ``Distributionally robust optimization under moment
  uncertainty with application to data-driven problems,'' \emph{Operations
  Research}, vol.~58, no.~3, pp. 595--612, 2010.

\bibitem{wang2023distributionally}
S.~Wang, ``Distributionally robust state estimation for jump linear systems,''
  \emph{IEEE Trans. Signal Processing}, 2023.

\bibitem{li2022tikhonov}
J.~Li, S.~Lin, J.~Blanchet, and V.~A. Nguyen, ``Tikhonov regularization is
  optimal transport robust under martingale constraints,'' \emph{Advances in
  Neural Information Processing Systems}, vol.~35, pp. 17\,677--17\,689, 2022.

\bibitem{kuhn2019wasserstein}
D.~Kuhn, P.~M. Esfahani, V.~A. Nguyen, and S.~Shafieezadeh-Abadeh,
  ``Wasserstein distributionally robust optimization: Theory and applications
  in machine learning,'' in \emph{Operations Research \& Management Science in
  the Age of Analytics}.\hskip 1em plus 0.5em minus 0.4em\relax Informs, 2019,
  pp. 130--166.

\bibitem{blanchet2019robust}
J.~Blanchet, Y.~Kang, and K.~Murthy, ``Robust {Wasserstein} profile inference
  and applications to machine learning,'' \emph{Journal of Applied
  Probability}, vol.~56, no.~3, pp. 830--857, 2019.

\bibitem{shorten2019survey}
C.~Shorten and T.~M. Khoshgoftaar, ``A survey on image data augmentation for
  deep learning,'' \emph{Journal of Big Data}, vol.~6, 2019.

\bibitem{saon2019sequence}
G.~Saon, Z.~T{\"u}ske, K.~Audhkhasi, and B.~Kingsbury, ``Sequence noise
  injected training for end-to-end speech recognition,'' in \emph{ICASSP
  2019-2019 IEEE International Conference on Acoustics, Speech and Signal
  Processing (ICASSP)}.\hskip 1em plus 0.5em minus 0.4em\relax IEEE, 2019, pp.
  6261--6265.

\bibitem{vu2015understanding}
K.~Vu, J.~C. Snyder, L.~Li, M.~Rupp, B.~F. Chen, T.~Khelif, K.-R. M{\"u}ller,
  and K.~Burke, ``Understanding kernel ridge regression: Common behaviors from
  simple functions to density functionals,'' \emph{International Journal of
  Quantum Chemistry}, vol. 115, no.~16, pp. 1115--1128, 2015.

\bibitem{wang1998robust}
X.~Wang and H.~V. Poor, ``Robust adaptive array for wireless communications,''
  \emph{IEEE J. Select. Areas in Commun.}, vol.~16, no.~8, pp. 1352--1366,
  1998.

\bibitem{katkovnik2006performance}
V.~Katkovnik, M.-S. Lee, and Y.-H. Kim, ``Performance study of the minimax
  robust phased array for wireless communications,'' \emph{IEEE Trans. Wireless
  Commun.}, vol.~54, no.~4, pp. 608--613, 2006.

\bibitem{rahimian2022frameworks}
H.~Rahimian and S.~Mehrotra, ``Frameworks and results in distributionally
  robust optimization,'' \emph{Open Journal of Mathematical Optimization},
  vol.~3, pp. 1--85, 2022.

\bibitem{kuhn2024distributionally}
D.~Kuhn, S.~Shafiee, and W.~Wiesemann, ``Distributionally robust
  optimization,'' \emph{Acta Numerica}, 2024.

\bibitem{sion1958general}
M.~Sion, ``{On general minimax theorems.}'' \emph{Pacific Journal of
  Mathematics}, vol.~8, no.~1, pp. 171 -- 176, 1958.

\end{thebibliography}

\begin{IEEEbiography}[{\includegraphics[width=1in,height=1.25in,clip,keepaspectratio]{wsx}}]{Shixiong Wang} (Member, IEEE) received the B.Eng. degree in detection, guidance, and control technology, and the M.Eng. degree in systems and control engineering from the School of Electronics and Information, Northwestern Polytechnical University, China, in 2016 and 2018, respectively. He received his Ph.D. degree from the Department of Industrial Systems Engineering and Management, National University of Singapore, Singapore, in 2022. 

He is currently a Postdoctoral Research Associate with the Intelligent Transmission and Processing Laboratory, Imperial College London, London, United Kingdom, from May 2023. He was a Postdoctoral Research Fellow with the Institute of Data Science, National University of Singapore, Singapore, from March 2022 to March 2023.

His research interest includes statistics and optimization theories with applications in signal processing (especially optimal estimation theory), machine learning (especially generalization error theory), and control technology.
\end{IEEEbiography}

\begin{IEEEbiography}
[{\includegraphics[width=1in,height=1.2in,clip,keepaspectratio]{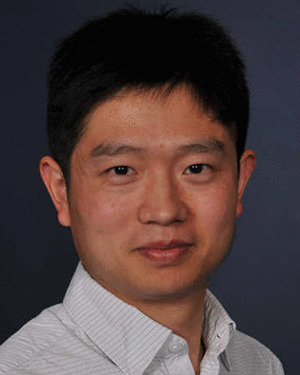}}] 
{Wei Dai} (Member, IEEE) received the Ph.D. degree from the University of Colorado Boulder, Boulder, Colorado, in 2007. He is currently a Senior Lecturer (Associate Professor) in the Department of Electrical and Electronic Engineering, Imperial College London, London, UK. From 2007 to 2011, he was a Postdoctoral Research Associate with the University of Illinois Urbana-Champaign, Champaign, IL, USA. His research interests include electromagnetic sensing, biomedical imaging, wireless communications, and information theory.
\end{IEEEbiography}

\begin{IEEEbiography}
[{\includegraphics[width=1in,height=1.2in,clip,keepaspectratio]{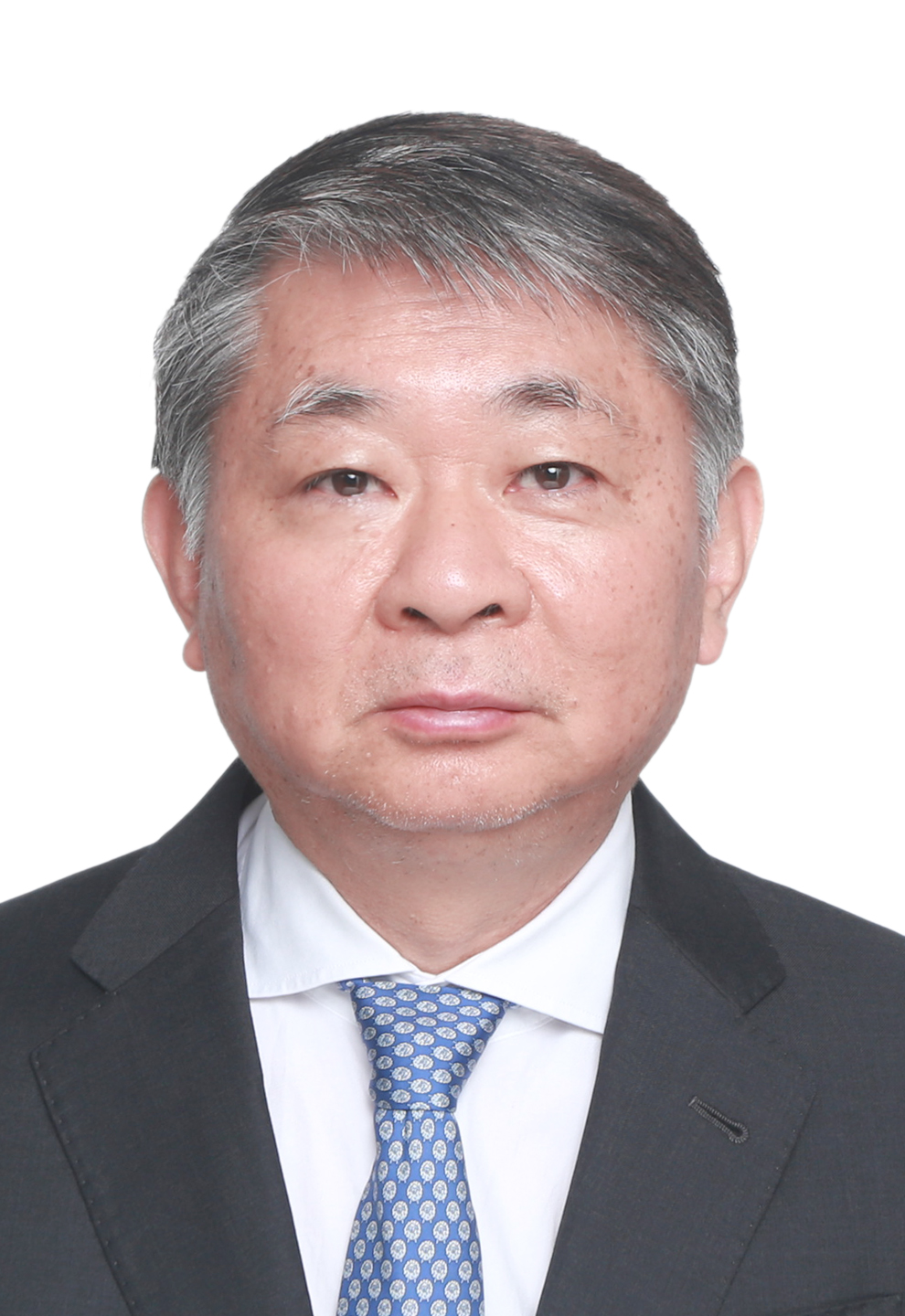}}] 
{Geoffrey Ye Li} is currently a Chair Professor at Imperial College London, UK. Before joining Imperial in 2020, he was a Professor at Georgia Institute of Technology for 20 years and a Principal Technical Staff Member with AT\&T Labs – Research (previous Bell Labs) for five years. He made fundamental contributions to orthogonal frequency division multiplexing (OFDM) for wireless communications, established a framework on resource cooperation in wireless networks, and introduced deep learning to communications. In these areas, he has published over 700 journal and conference papers in addition to over 40 granted patents. His publications have been cited around 80,000 times with an H-index over 130. He has been listed as a Highly Cited Researcher by Clarivate/Web of Science almost every year.

Dr. Geoffrey Ye Li was elected to Fellow of the Royal Academic of Engineering (FREng), IEEE Fellow, and IET Fellow for his contributions to signal processing for wireless communications. He received 2024 IEEE Eric E. Sumner Award, 2019 IEEE ComSoc Edwin Howard Armstrong Achievement Award, and several other awards from IEEE Signal Processing, Vehicular Technology, and Communications Societies.
\end{IEEEbiography}


\newpage

\setcounter{page}{1}
\begin{figure*}
\centering
{\large \bf Supplementary Materials}
\end{figure*}

\section{Additional Experimental Results}\label{append:expt-scenario-1}
Complementary to experimental setups in Section \ref{sec:experiment}, we consider pure complex Gaussian channel noises. First, we suppose that the transmit antennas emit continuous-valued complex signals; without loss of generality, Gaussian signals are used in experiments. The performance evaluation measure is therefore the mean-squared error (MSE). The experimental results are shown in Fig. \ref{fig:scenario-1-1}.

\begin{figure}[!htbp]
    \centering
    \subfigure[$N$=8, SNR 10dB, $\mat R_v$ Estimated]{
        \begin{minipage}[htbp]{0.46\linewidth}\label{fig-Gaussian:a}
            \centering
            \includegraphics[height=3.4cm]{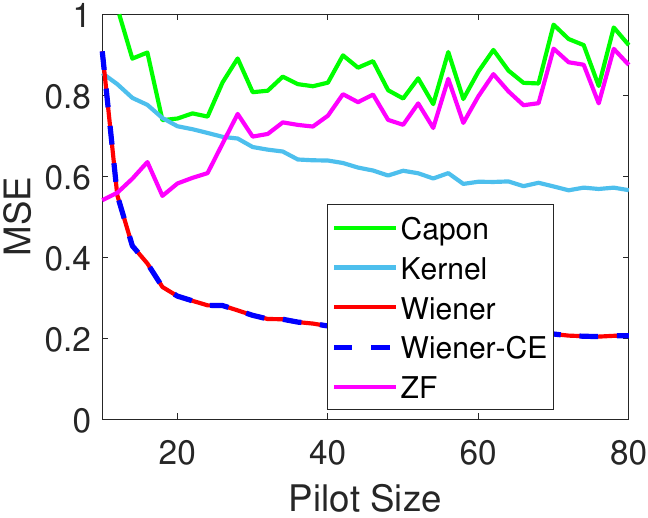}
        \end{minipage}
    }
    \subfigure[$N$=8, SNR 10dB, $\mat R_v$ Known]{
        \begin{minipage}[htbp]{0.46\linewidth}\label{fig-Gaussian:b}
            \centering
            \includegraphics[height=3.4cm]{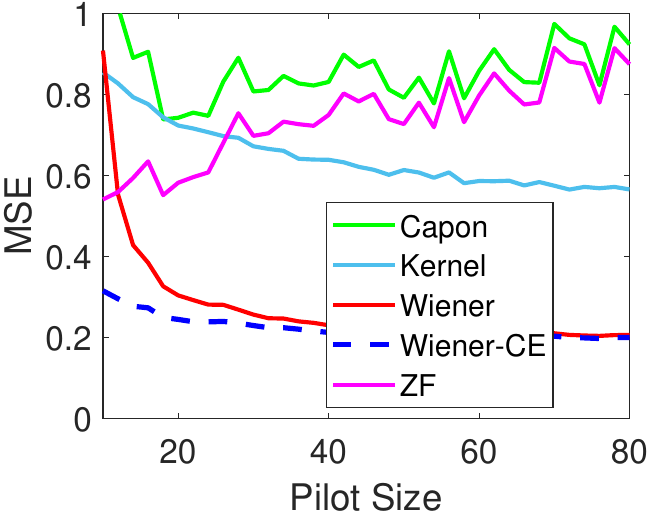}
        \end{minipage}
    }
    
    \subfigure[$N$=16, SNR 10dB, $\mat R_v$ Estimated]{
        \begin{minipage}[htbp]{0.46\linewidth}\label{fig-Gaussian:c}
            \centering
            \includegraphics[height=3.4cm]{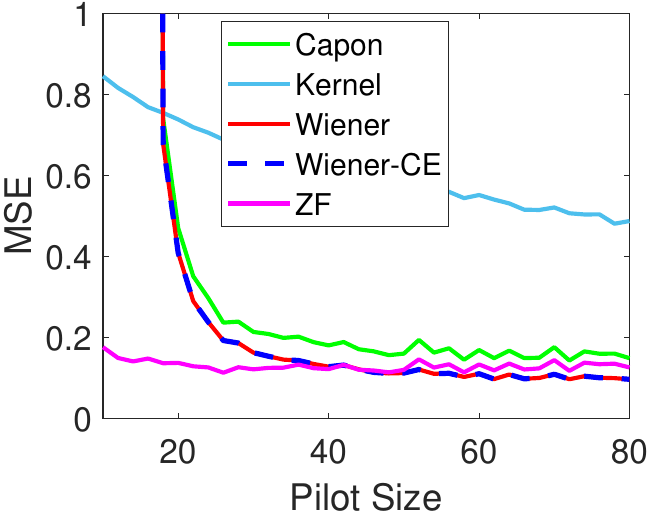}
        \end{minipage}
    }
    \subfigure[$N$=16, SNR -10dB, $\mat R_v$ Estimated]{
        \begin{minipage}[htbp]{0.46\linewidth}\label{fig-Gaussian:d}
            \centering
            \includegraphics[height=3.4cm]{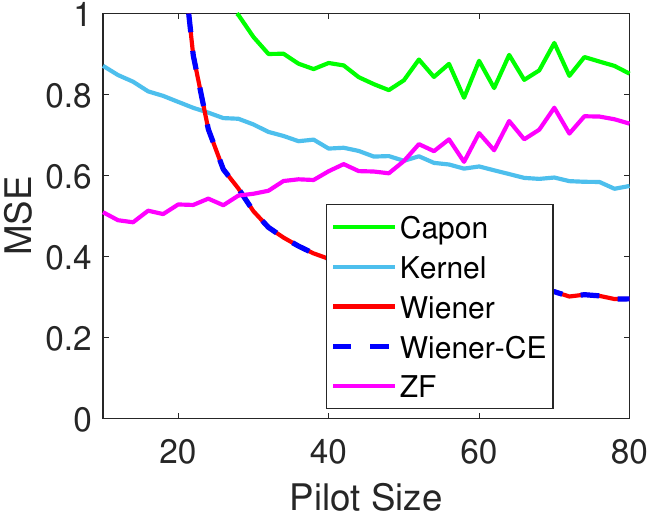}
        \end{minipage}
    }
    
    \caption{Testing MSE against training pilot sizes under different numbers of receive antennas; only non-robust beamformers including non-diagonal-loading ones are considered. The true value of $\mat R_v$ can be unknown and estimated using pilot data. The signal-to-noise ratio (SNR) is $10$dB or $-10$dB.}
    \label{fig:scenario-1-1}
\end{figure}

From Fig. \ref{fig:scenario-1-1}, the following main points can be outlined.
\begin{enumerate}
    \item For a fixed number $M$ of transmit antennas, the larger the number $N$ of receive antennas, the smaller the MSE; cf. Figs. \ref{fig-Gaussian:a} and \ref{fig-Gaussian:c}. This fact is well-established and is due to the benefit of antenna diversity. In addition, for fixed $N$ and $M$, the higher the SNR, the smaller the MSE; cf. Figs. \ref{fig-Gaussian:c} and \ref{fig-Gaussian:d}; this is also well believed.
    
    \item As the pilot size increases, the Wiener beamformer tends to have the best performance because the Wiener beamformer is optimal for the linear Gaussian signal model. When $\mat R_v$ is accurately known, the Wiener-CE beamformer outperforms the general Wiener beamformer (cf. Fig. \ref{fig-Gaussian:b}) because the former also exploits the information of the linear signal model in addition to the pilot data, while the latter only utilizes the pilot data. However, when $\mat R_v$ is estimated using the pilot data, the performances of the general Wiener beamformer and the Wiener-CE beamformer have no significant difference; cf. Figs. \ref{fig-Gaussian:a} and \ref{fig-Gaussian:c}. Therefore, Fig. \ref{fig:scenario-1-1} validates our claim that \textit{channel estimation is not a necessary operation in receive beamforming and estimation of wireless signals}; recall Subsection \ref{subsubsec:role-CE}.

    \item The ZF beamformer tends to be more efficient as $N$ increases; cf. Figs. \ref{fig-Gaussian:a} and \ref{fig-Gaussian:c}. However, the ZF beamformer becomes less satisfactory when the SNR decreases; cf. Figs. \ref{fig-Gaussian:c} and \ref{fig-Gaussian:d}. The Capon beamformer is also unsatisfactory when $N$ is small or the SNR is low.

    \item The kernel beamformer, as a nonlinear method, cannot outperform linear beamformers because, for a linear Gaussian signal model, the optimal beamformer is linear. From the perspective of machine learning, nonlinear methods tend to overfit the limited training samples.
\end{enumerate}

Second, we suppose that the transmit antennas emit discrete-valued symbols from a constellation that is modulated using quadrature phase-shift keying (QPSK). The performance evaluation measure is therefore the symbol error rate (SER). The experimental results are shown in Fig. \ref{fig:scenario-1-2}. We find that all the conclusive main points from Fig. \ref{fig:scenario-1-1} can be obtained from Fig. \ref{fig:scenario-1-2} as well: this validates that \textit{minimizing MSE reduces SER}. In addition, Figs. \ref{fig-QPSK:c} and \ref{fig-QPSK:d} reveal that the Wiener beamformer even slightly works better than the Wiener-CE beamformer when the pilot size is smaller than $15$ because the uncertainty in the estimated $\math R_v$, on the contrary, misleads the latter. Nevertheless, as the pilot size increases, the Wiener-CE beamformer tends to overlap the Wiener beamformer quickly. 
\begin{figure}[!htbp]
    \centering
    \subfigure[$N$=8, SNR 10dB, $\mat R_v$ Estimated]{
        \begin{minipage}[htbp]{0.46\linewidth}\label{fig-QPSK:a}
            \centering
            \includegraphics[height=3.4cm]{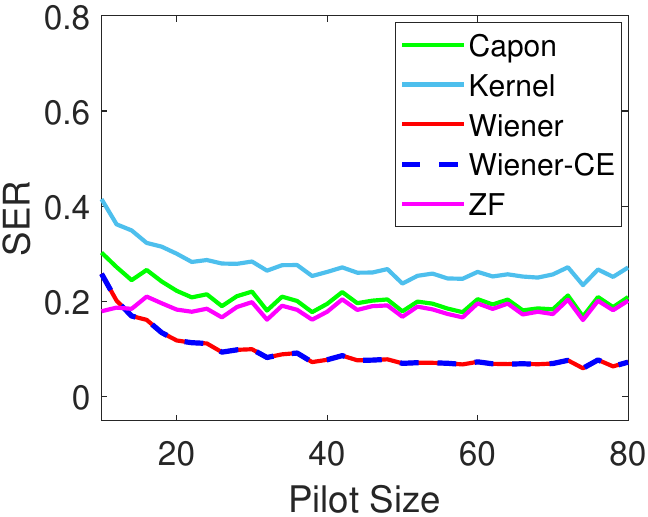}
        \end{minipage}
    }
    \subfigure[$N$=8, SNR 10dB, $\mat R_v$ Known]{
        \begin{minipage}[htbp]{0.46\linewidth}\label{fig-QPSK:b}
            \centering
            \includegraphics[height=3.4cm]{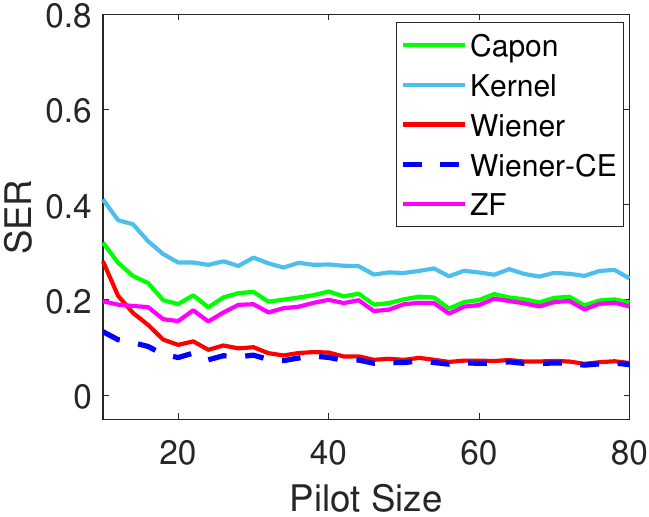}
        \end{minipage}
    }
    
    \subfigure[$N$=16, SNR 10dB, $\mat R_v$ Estimated]{
        \begin{minipage}[htbp]{0.46\linewidth}\label{fig-QPSK:c}
            \centering
            \includegraphics[height=3.4cm]{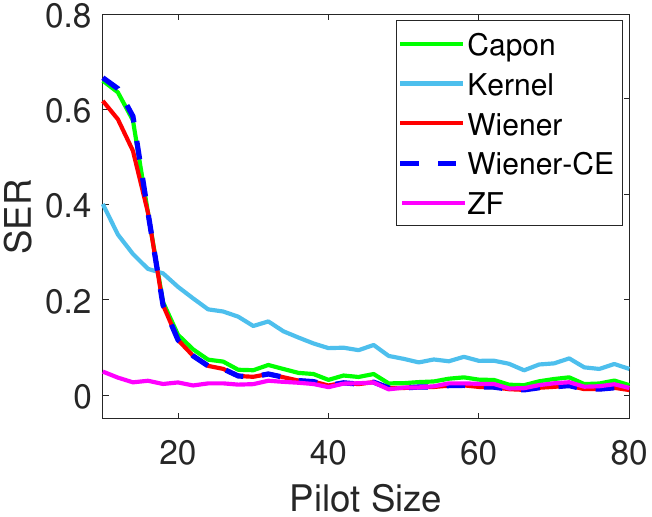}
        \end{minipage}
    }
    \subfigure[$N$=16, SNR -10dB, $\mat R_v$ Estimated]{
        \begin{minipage}[htbp]{0.46\linewidth}\label{fig-QPSK:d}
            \centering
            \includegraphics[height=3.4cm]{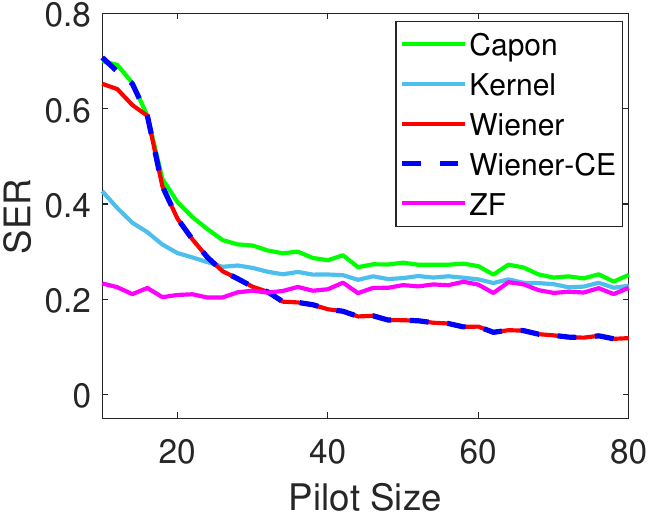}
        \end{minipage}
    }
    
    \caption{Testing SER against training pilot sizes under different numbers of receive antennas; only non-robust beamformers including non-diagonal-loading ones are considered. The true value of $\mat R_v$ can be unknown and estimated using pilot data. The signal-to-noise ratio (SNR) is $10$dB or $-10$dB.}
    \label{fig:scenario-1-2}
\end{figure}

\end{document}